\documentclass[a4paper,twocolumn,11pt,accepted=2023-01-12]{quantumarticle}
\pdfoutput=1
\usepackage[english]{babel}
\usepackage{amsmath}

\usepackage{tikz}
\usepackage{lipsum}

\usepackage{hyperref}
\usepackage{fontenc}
\usepackage[latin9]{inputenc}
\usepackage{amsmath}
\usepackage{amsthm}
\usepackage{amssymb}
\usepackage{stmaryrd}
\usepackage{graphicx}
\usepackage{fullpage}
\usepackage{xspace}
\usepackage{multirow}
\usepackage{enumitem}
\usepackage{rotating}
\usepackage{mathrsfs}
\usepackage{mathtools}
\usepackage{braket}
\usepackage{upgreek}
\usepackage{thmtools}
\usepackage{thm-restate}
\usepackage{comment}
\usepackage[normalem]{ulem}
\usepackage{caption}
\declaretheorem[name=Theorem]{thm}

\makeatletter
\theoremstyle{definition}
\newtheorem{defn}{\protect\definitionname}
\theoremstyle{plain}
\newtheorem{lem}{\protect\lemmaname}
\theoremstyle{plain}

\theoremstyle{plain}

\DeclareMathOperator{\Tr}{Tr}
\DeclareMathOperator{\Prob}{Pr}
\DeclareMathOperator{\I}{I}

\usepackage{babel}
\providecommand{\corollaryname}{Corollary}
\providecommand{\definitionname}{Definition}
\providecommand{\lemmaname}{Lemma}
\providecommand{\propositionname}{Proposition}

\definecolor{darkgreen}{rgb}{0,0.5,0}

\newcommand{\C}{\mathbb{C}}

\newcommand{\A}{\mathbf{a}}
\newcommand{\X}{\mathbf{x}}

\newcommand{\cH}{\mathcal{H}}
\newcommand{\cS}{\mathcal{S}}
\newcommand{\cR}{\mathcal{R}}
\newcommand{\PH}{\mathcal{P}(\mathcal{H})}
\renewcommand{\A}{\mathbf{a}}
\renewcommand{\X}{\mathbf{x}}

\DeclarePairedDelimiter\abs{\lvert}{\rvert}%
\DeclarePairedDelimiter\norm{\lVert}{\rVert}%

\makeatletter
\let\oldabs\abs
\def\abs{\@ifstar{\oldabs}{\oldabs*}}

\let\oldnorm\norm
\def\norm{\@ifstar{\oldnorm}{\oldnorm*}}
\makeatother

\begin{document}

\title{Contextuality in composite systems:\newline the role of entanglement in the Kochen--Specker theorem}
\author{Victoria J~Wright}
\email{victoria.wright@icfo.eu}
\affiliation{ICFO-Institut de Ciencies Fotoniques, The Barcelona Institute of Science and Technology, 08860 Castelldefels, Spain}

\author{Ravi Kunjwal}
\email{ravi.kunjwal@ulb.be}
\affiliation{Centre for Quantum Information and Communication, Ecole polytechnique de Bruxelles,
	CP 165, Universit\'e libre de Bruxelles, 1050 Brussels, Belgium}

\begin{abstract}
The Kochen--Specker (KS) theorem reveals the nonclassicality of single quantum systems. In contrast, Bell's theorem and entanglement concern the nonclassicality of composite quantum systems. Accordingly, unlike incompatibility, entanglement and Bell non-locality are not necessary to demonstrate KS-contextuality. However, here we find that for multiqubit systems, entanglement and non-locality are \emph{both} essential to proofs of the Kochen--Specker theorem. Firstly, we show that unentangled measurements (a strict superset of local measurements) can never yield a logical (state-independent) proof of the KS theorem for multiqubit systems. In particular, unentangled but nonlocal measurements---whose eigenstates exhibit ``nonlocality without entanglement"---are insufficient for such proofs.
This also implies that proving Gleason's theorem on a multiqubit system necessarily requires entangled projections, as shown by Wallach \href{http://dx.doi.org/10.1090/conm/305/05226}{[Contemp Math, 305: 291-298 (2002)]}. 
 Secondly, we show that a multiqubit state admits a statistical (state-dependent) proof of the KS theorem if and only if it can violate a Bell inequality with projective measurements. 
We also establish the relationship between entanglement and the theorems of Kochen--Specker and Gleason more generally in multiqudit systems by constructing new examples of KS sets. Finally, we discuss how  our results shed new light on the role of multiqubit contextuality as a resource within the paradigm of quantum computation with state injection.
\end{abstract}

\maketitle


\section{Introduction}
Quantum theory's `departure from classical lines of thought' \cite{Schrodinger35} is today a driving force behind the promise of quantum technologies. Quantum theory abandons assumptions implicit in classical physical theories, such as the assumptions that physics must be fundamentally deterministic or that all observables are jointly measurable. This allows for an array of typically quantum phenomena such as entanglement, uncertainty relations, Bell nonlocality, contextuality, etc.~We generically refer to the possibility of these previously forbidden properties as the \emph{nonclassicality} of quantum theory. Such nonclassical properties of quantum theory are key to the many advantages quantum information processing and quantum computation hold over their classical counterparts. However, the exact  relationship between different notions of nonclassicality and such advantages is often unclear and remains an active area of research \cite{LP01, DG07, VFG12,HWV14,CHT20, CLM10, YK20, FBL21}.

The case of multiqubit systems is of particular importance given their ubiquity throughout quantum technologies, particularly noisy intermediate-scale quantum (NISQ) technologies \cite{Preskill18, AAB19}. However, the nonclassicality of qubit systems remains an anomalous case. Individually, for example, a qubit cannot display Kochen--Specker (KS) contextuality \cite{KS67} while, collectively, multiqubit systems (which exhibit KS-contextuality) derail the neat narrative of such contextuality powering quantum computational advantage \cite{HWV14,BDB17}. Furthermore, entanglement is often considered a key indicator of nonclassicality in these systems but the sense in which it relates to fundamental notions of nonclassicality witnessed by the theorems of Bell \cite{Bell64, Bell66}, Kochen--Specker \cite{KS67}, and Gleason \cite{Gleason57} needs clarification. Since the question of nonclassicality is essentially a foundational one \cite{Spekkens16}, we approach the study of multiqubit systems through this lens.

Entanglement is an intrinsically compositional property and is, therefore, only relevant to the study of nonclassicality in composite (i.e. multipartite) systems. Schr\"{o}dinger claimed the entanglement of quantum states to be `the characteristic trait of quantum mechanics, the one that enforces its entire departure from classical lines of thought' \cite{Schrodinger35}. Bell's theorem \cite{Bell64, Bell66} exemplifies this point by revealing the nonclassicality of correlations arising from local measurements on composite quantum systems in entangled states, the simplest case being a two-qubit system. 

The Kochen--Specker theorem \cite{KS67}, on the other hand, can reveal the nonclassicality of correlations between measurements implemented on an \emph{indivisible} quantum system, the simplest case being a qutrit. For quantum systems of dimension at least three, the theorem states that there cannot exist an underlying ontological model---known as a  \emph{KS-noncontextual} ontological model---that reproduces the predictions of quantum theory. In such a model the outcomes of projective measurements are fully determined by the ontic state of the system. Furthermore, the outcomes are independent of \emph{context}\footnote{The context  of a projective measurement here refers to other projective measurements with which it is jointly performed.} and respect the functional relationships between commuting measurements.\footnote{For example, the outcome of measuring an observable with operator $A^2$ should be the square of the outcome of measuring $A$.}  

The key insight of Kochen and Specker was that one can witness the impossibility of such models with a {\em finite} set of rank-$1$ projections  (a \emph{KS set}) on a three-dimensional Hilbert space. A KS set thus constitutes a {\em logical} proof of the KS theorem \cite{KS15}, exemplified by the original result of Kochen and Specker \cite{KS67}. It is also possible to demonstrate the inadequacy (rather than impossibility) of KS-noncontextual ontological models in a (weaker) {\em statistical} sense. Such statistical proofs of the KS theorem  \cite{KS18} are exemplified by the proof due to Klyachko {\em et al.}~\cite{KCBS}.

Prior to the Bell and Kochen--Specker theorems, Gleason's theorem \cite{Gleason57} demonstrated that, for any quantum system of dimension at least three, the unique way to assign probabilities to the outcomes of projective measurements is via the Born rule.  In particular, Gleason's theorem excludes any deterministic probability rule  given by a $\{0,1\}$-valued assignment of probabilities to all the self-adjoint projections on the system's Hilbert space.~This exclusion thus implies the KS theorem, but, unlike the proof of Kochen and Specker \cite{KS67}, it requires an uncountably infinite KS set.

A single qubit does not support any of the three theorems  (Bell, KS, or Gleason) and is, by that token, rather ``classical''.\footnote{This assumes a restriction to pure states and projective measurements on a qubit, as is traditionally the case in these theorems. Such a `pure' qubit can exhibit nonclassicality in other (weaker) respects like the existence of incompatible measurements, something reproducible in classical theories with an epistemic restriction \cite{Spekkens16}. However, once mixed states and generalised measurements are included, a single qubit can support proofs of generalised contextuality \cite{Spekkens05, KG14, Kunjwal16}. Allowing generalised measurements also permits Gleason-type theorems for the qubit case \cite{Busch03, CFM04, WW19}.} A single qutrit can support the Gleason and KS theorems but not Bell's theorem. Hence, the smallest quantum system on which one can meaningfully study the  interplay of Gleason, Bell, and KS theorems is a two-qubit system. The nonclassicality witnessed by Bell's theorem in qubit pairs (and more generally) clearly depends on entanglement, since Bell inequality violations can only be observed when the quantum systems used are described by an entangled state. Is entanglement, however, also necessary for the theorems of KS and Gleason in a two-qubit system, and more generally, in multiqubit systems? Since both theorems appeal to the structure of quantum  measurements, we need to go beyond states and consider the role of entanglement in measurements as well.

The question of entanglement and Gleason's theorem is already resolved. A result by Wallach \cite{Wallach02} showed that the set of unentangled multiqudit projections yields a proof of Gleason's theorem if and only if each qudit has a Hilbert space dimension three or more. In particular, the set of unentangled multiqubit projections cannot yield Gleason's theorem.

This work will largely consider rank-one projective measurements that are \emph{unentangled}, i.e. measurements in bases comprising only product vectors.  Unlike quantum states, measurements can be ``nonlocal" without being entangled, that is, there exist measurements which cannot be implemented via local operations and classical communication (LOCC) but nevertheless only involve projections onto unentangled subspaces.\footnote{The  eigenstates for such a measurement exhibit a phenomenon called `nonlocality without entanglement', i.e.~they form a set of mutually orthogonal states that cannot be distinguished perfectly via LOCC \cite{BDF99}.} For example, one cannot perform a measurement in the unentangled three qubit basis 
\begin{equation}\label{eq:nlbasis}
\begin{aligned}
\{&\ket{000},\ket{+10},\ket{0+1},\ket{10+},\\&\ket{111},\ket{-10},\ket{0-1},\ket{10-}\}\,,
\end{aligned}
\end{equation}
via LOCC. It could in principle be that ``nonlocality" of unentangled measurements is sufficient to recover a logical proof of the KS theorem \cite{KS67, KS15}. However, in the first main result of this work, Theorem \ref{thm:main}, we show that this is not the case: the unentangled rays of a multiqubit system are KS-colourable and, thus, no unentangled form of nonclassicality suffices for a logical proof of KS-contextuality in this setting.

\begin{figure}
	\centering
	\begin{tabular}{|c|c|c|}
		\hline
		$X\otimes I$ & $I\otimes X$ & $X\otimes X$\\
		\hline
		$I\otimes Y$ & $Y\otimes I$ & $Y\otimes Y$\\
		\hline
		$X\otimes Y$ & $Y\otimes X$ & $Z\otimes Z$\\
		\hline
	\end{tabular}\vspace{0.3cm}
	\includegraphics[scale=0.3]{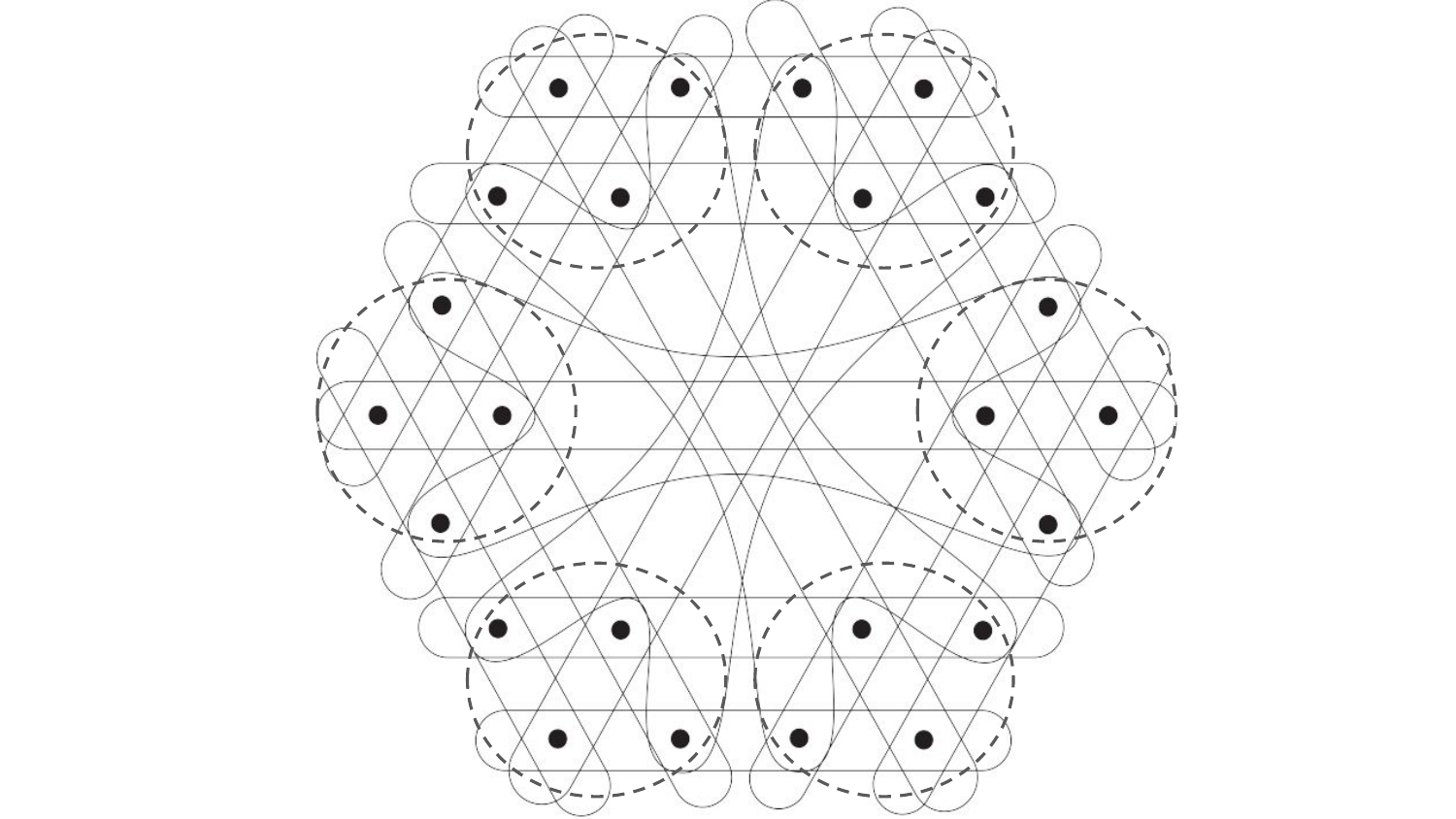}
	\caption{The Peres-Mermin square \cite{Mermin93} consisting of two-qubit Pauli matrices together with the contextuality scenario defined by the orthogonality relations between Peres's 24 rays \cite{Peres91}. Each row or column of the Peres-Mermin square is associated to a two-qubit orthonormal basis in which all the operators in that row or column are diagonal. In the contextuality scenario, the six dashed hyperedges denote the six orthonormal bases corresponding to the rows and columns of the Peres-Mermin square. In particular, the basis that diagonalises the third column, $\{XX,YY,ZZ\}$, is the Bell basis.}
	\label{peresmermin}
\end{figure}

The well-known proof of the KS theorem via the Peres-Mermin square \cite{Peres90, Peres91, Mermin93} at first appears to not involve entanglement (see Fig.~\ref{peresmermin}). However, the joint measurement of the Pauli observables $\{XX,YY,ZZ\}$, for example, necessitates a measurement in the Bell basis. In this sense, our result shows that this entanglement is not accidental but, rather, unavoidable; any logical (hence, state-independent) \cite{KS15} proof of  the KS theorem for multiqubit systems necessarily requires entangled measurements. 

We also construct a KS-noncontextual ontological model for the fragment of multiqubit quantum theory containing unentangled measurements and product states. The model can be viewed as a generalization of the single-qubit model of Kochen and Specker \cite{KS67}, but the proof of its validity relies on our result---Lemma \ref{lem:north}---to show that the ontic states are indeed valid KS-noncontextual ontic states. The model admits a simple extension to the case of separable states which renders it preparation contextual \cite{Spekkens05}, but still KS-noncontextual. 

The existence of this model implies that for a proof of the KS theorem on a multiqubit system, one requires either (i) entanglement in the measurements, in which case one can provide a logical (and {\em state-independent}) proof, or (ii) entanglement in the state, in which case the violation of Bell inequalities, for example, provides a \emph{state-dependent} proof  without any entangled measurements \cite{AFL15}. In the second main result of this paper, we demonstrate that such Bell inequality violations are the \emph{only} way to prove the KS theorem in this setting. We show that an entangled state enables a (finite) statistical proof of the KS theorem with unentangled measurements if and only if it also violates a Bell inequality with local projective measurements.

Thus, we must conclude that, just as in the case of Bell's and Gleason's theorems, the nonclassicality of multiqubit systems is underpinned by entanglement in the case of the Kochen--Specker theorem. This discovery is in surprising contrast to the usual intuition that takes the KS theorem as witnessing nonclassicality that is independent of entanglement because it applies, for example, to a single qutrit \cite{KS67}. Hence, in the simplest case where {\em all} three theorems apply, i.e.~a two-qubit system, entanglement is necessary for all of them.

Exploring the relationship between the KS theorem and entanglement in multiqudit systems further, we provide two new constructions of KS sets on multiqudit systems that allow us to obtain an overall picture of this relationship (Fig.~\ref{unentksgleason}). As displayed in Fig.~\ref{unentksgleason}, our results recover the result of Wallach \cite{Wallach02} in the multiqubit case as a corollary, namely, that entangled projections are necessary to obtain Gleason's theorem.
\begin{figure}
	\includegraphics[width=\columnwidth]{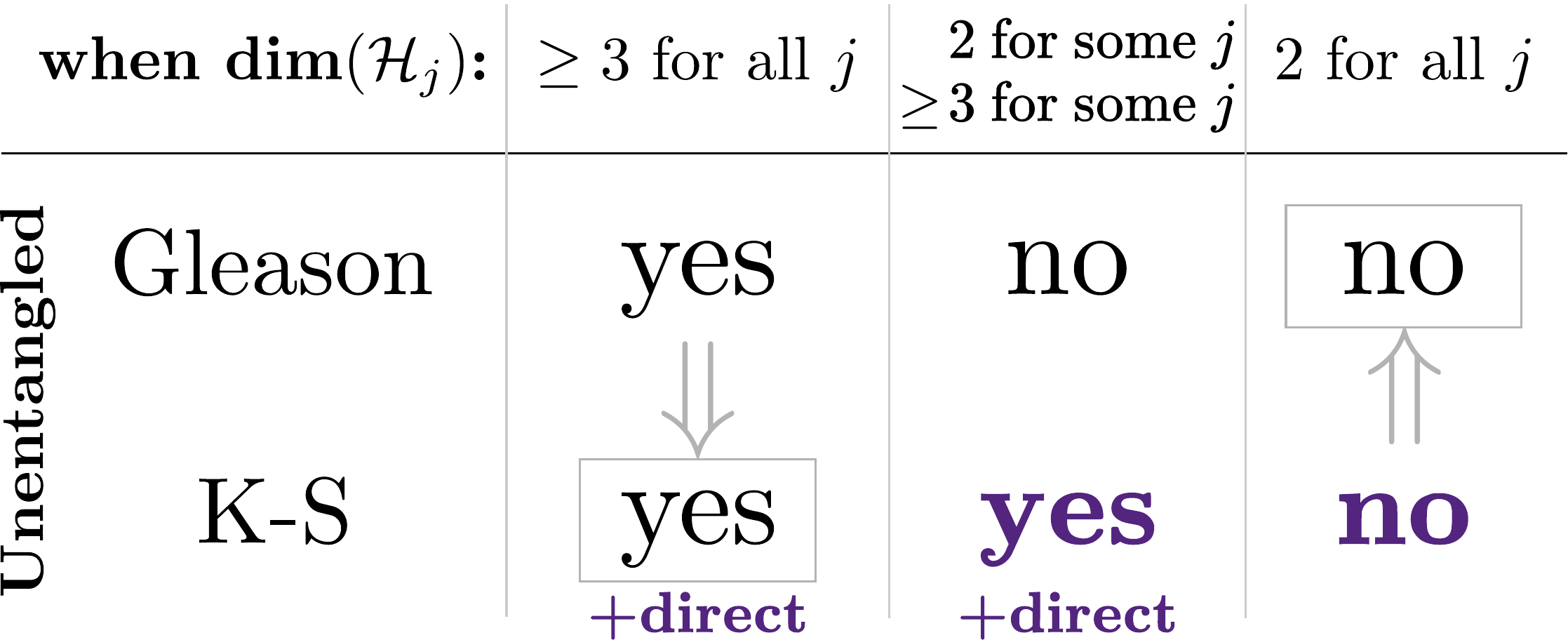}
	\caption{The existence of logical proofs of the KS theorem and proofs of Gleason's theorem without entanglement on a Hilbert space $\cH_1\otimes\cdots\otimes\cH_n$, where $1\leq j\leq n$. The implication arrows show which results follow from each other. The addition of +direct signifies the result also holds for the subset of unentangled measurements given by direct products bases, see Sec.~\ref{sec:unentks}. The results in bold and purple are introduced in the present work.}\label{unentksgleason}
\end{figure}

Finally, we discuss implications of our results for the role of contextuality in multiqubit quantum computation with state injection. In particular, we highlight how the choice of measurements in the schemes of Ref.~\cite{BDB17} appear natural in view of our results.

The structure of the paper is as follows. In Sec.~\ref{sec:prelims}, we provide preliminary notions that we will need in the rest of the paper. In Sec.~\ref{sec:mainresult}, we prove our first main result, i.e.~the necessity of entanglement for multiqubit KS sets. Sec.~\ref{sec:model} presents a KS-noncontextual ontological model for product states and unentangled measurements of multiqubit systems. Sec.~\ref{sec:bellks} establishes our second main result: an entangled state yields a statistical proof of the KS theorem if and only if it yields a proof of Bell's theorem. In Sec.~\ref{sec:fullyent}, we construct a two-qubit KS set without any fully entangled bases. Sec.~\ref{sec:unentks} proves the existence of unentangled KS sets on any multiqudit system that contains at least one qudit of Hilbert space dimension three or more. Sec.~\ref{sec:qcomp} discusses implications of our results for the role of contextuality in multiqubit schemes for quantum computation with state injection (QCSI). Sec.~\ref{sec:disc} concludes with a discussion of our results.

\section{Preliminaries}\label{sec:prelims}

Given a separable Hilbert space, $\mathcal{H}$, a self-adjoint projection on $\mathcal{H}$ is a linear operator $\Pi$ satisfying $\Pi^2=\Pi^\dagger=\Pi$, where $\Pi^\dagger$ denotes the adjoint of $\Pi$. We will denote the set of self-adjoint projections on $\mathcal{H}$ by $\mathcal{P}(\mathcal{H})$. In this work, we will consider \emph{projective measurements} on quantum systems, which we describe by sets $\left\{\Pi_1,\Pi_2,\ldots\right\}$ of self-adjoint, mutually orthogonal projections $\Pi_j\in\PH$ on a separable Hilbert space, $\mathcal{H}$, that sum to the identity operator. 

Rank-one projections are generally sufficient for our purposes so we equivalently consider rays in the projective Hilbert space $\mathcal{R}(\mathcal{H})$. The projective Hilbert space $\mathcal{R}(\mathcal{H})$ is given by the set of equivalence classes, or \emph{rays}, of non-zero vectors in $\mathcal{H}$, under the equivalence relation $\psi\sim\chi$ if and only if $\psi=\alpha\chi$ for some non-zero $\alpha\in\mathbb{C}$. We represent each ray with one of its unit vectors, which we denote by a ``ket'', such as $\ket{\psi}$. For brevity, we will often refer to this vector as the ray it represents. 

There is a bijection between the sets of rays and rank-one projections on a Hilbert space. A rank-one projective measurement is therefore uniquely defined by a \emph{complete set} of rays, i.e.~a set of $d$ mutually orthogonal rays, where $d$ is the (possibly infinite) dimension of the Hilbert space. A complete set is represented by an orthonormal basis $\{\ket{\psi_1},\ket{\psi_2},\ldots\}$ of $\cH$, and thus we say we are performing a measurement \emph{in a basis} and often refer to a complete set of rays as simply a basis.

The KS theorem 
states the impossibility of an outcome-deterministic and measurement noncontextual ontological model  \cite{Spekkens05} (briefly, a {\em KS-noncontextual} ontological model) for quantum theory. We will refer to this fact as the KS-contextuality of quantum theory. The original result due to Kochen and Specker \cite{KS67} provides a \emph{logical} proof of the KS theorem \cite{KS15} and we state it below in a formulation most relevant to the present work:
\begin{thm}[Kochen--Specker \cite{KS67}]\label{thm:KS}
Given a separable Hilbert space $\mathcal{H}$ of dimension at least three, there does not exist a map
\begin{equation}
c:\mathcal{R}\rightarrow\{0,1\}
\end{equation}
such that for any complete set of rays $\{\ket{\psi_1},\ket{\psi_2},\ldots\}$, we have $c(\ket{\psi_j})=1$ for exactly one value of $j\in\{1,2,\ldots\}$.
\end{thm}
This formulation of the result implies the traditional statement of the Kochen--Specker theorem in terms of \emph{valuation} functions on self-adjoint operators (see Appendix \ref{app:trad}). The KS theorem also admits \emph{statistical} proofs, whereby it is shown that the statistics produced by a collection of measurements performed on a fixed quantum state cannot be derived from a KS-noncontextual ontological model \cite{KCBS, KS18}.\footnote{The notion of contextuality was extended beyond the Kochen--Specker notion to a generalised notion of contextuality in Ref.~\cite{Spekkens05}. KS-noncontextuality, within this generalised framework, is recovered as a conjunction of two assumptions on ontological models of any operational theory: measurement noncontextuality and outcome determinism for projective measurements. As we will not use much of the machinery of generalised contextuality in this paper (and since we are focusing on quantum theory rather than operational theories in general), we refer the interested reader to Refs.~\cite{Spekkens05, KS15, Kunjwal16,KS18, Kunjwal19, Kunjwal20} for discussions of generalised contextuality and its connection with KS-contextuality.} We will discuss these proofs of the KS theorem, and how they differ from logical proofs, more extensively in Sec.~\ref{sec:bellks}.

A set of rays in a Hilbert space can be represented by a hypergraph in which there is a vertex for each ray and each complete set of rays constitutes a hyperedge. Such a hypergraph is an example of a \emph{contextuality scenario} \cite{AFL15}. The map $c$ in Theorem \ref{thm:KS} then defines a special type of 2-colouring of this hypergraph, which we will call a \emph{KS-colouring}.\footnote{Here ``0'' and ``1" stand in for two possible ``colours" that could be assigned to vertices according to the rule specified in Theorem \ref{thm:KS}.} A set of rays for which there does not exist a KS-colouring is known as a \emph{KS set}.

A KS-colouring of a set of rays is more than simply a mathematical tool: KS-colourings define the ontic states in a KS-noncontextual ontological model. The ontic state determines the outcome of any measurement with certainty, and this choice of one deterministic outcome from each measurement is exactly a KS-colouring. The existence of a KS set in any Hilbert space of dimension greater than two then proves the KS theorem (via Theorem \ref{thm:KS})

The fact that all the rays in a Hilbert space, $\cH$, of dimension three or more form a KS set (i.e., a proof of the KS theorem) follows from Gleason's theorem. 
Kochen and Specker \cite{KS67}, however, gave an explicit construction of a \textit{finite} KS set in a three-dimensional Hilbert space, arguably providing a much simpler proof of the KS theorem than relying on Gleason's theorem.\footnote{Hrushovski and Pitowsky \cite{HP04}, however, showed that Gleason's theorem, when combined with the compactness theorem of first-order logic, implies the \textit{existence} of a finite KS set. Hence, Gleason's theorem implies not only the KS theorem but something stronger, i.e., the existence of finite KS sets.}

\begin{thm}[Gleason \cite{Gleason57}]
Let $\cH$ be a separable Hilbert space of dimension at least three. Any map $f:\mathcal{P}(\cH)\rightarrow[0,1]$ satisfying
\begin{equation}
f(\Pi_1)+f(\Pi_2)+\cdots=f(\Pi_1+\Pi_2+\cdots)\,,
\end{equation}
for any set of mutually orthogonal projections $\left\{\Pi_1,\Pi_2,\ldots\right\}$, and $f(\I_\cH)=1$ where $\I_\cH$ is the identity operator on $\cH$, admits an expression
\begin{equation}\label{bornrule}
f(\Pi)=\Tr(\Pi\rho)\,,
\end{equation} 
 for some density operator $\rho$ on $\cH$.
\end{thm}

The maps $f$ in Gleason's theorem are known as \emph{frame functions}. Any KS-colouring, $c$, on the rays (and equivalently the rank-one projections) of a Hilbert space $\cH$ would extend to a $\{0,1\}$-valued frame function on $\PH$ via $c(\Pi_1+\Pi_2+\cdots)\equiv c(\Pi_1)+c(\Pi_2)+\cdots$ for mutually orthogonal sets of rank-one projections $\{\Pi_1,\Pi_2,\ldots\}$. In dimensions greater than two, however, Gleason's theorem shows that such a frame function does not exist, hence the set of all rays is not KS-colourable.

Neither the KS theorem nor Gleason's theorem hold for a single two-level system, i.e.~a qubit. The theorems do, however, hold for systems of multiple qubits. We will examine whether the onset of the applicability of these theorems is due to the presence entanglement in the rays of the measurements or if it is due to a weaker notion of nonlocality (without entanglement) \cite{BDF99}. Specifically, if we describe a projective measurement on a composite Hilbert space $\cH_1\otimes\cdots\otimes\cH_n$ by a sequence of projections $\Pi_1,\Pi_2,\ldots$, we say the measurement is \emph{unentangled} if the support of $\Pi_k$ admits a basis 
of \emph{product} vectors, i.e.~vectors $\psi=\psi_1\otimes\cdots\otimes\psi_n \in \cH$, where $\psi_j\in\cH_j$ for all $j$. Rank-one unentangled measurements are sufficient for our argument. Each such measurement can be described by a complete set of \emph{product rays}, that is, rays consisting of product vectors.

In the case of Gleason's theorem, Wallach \cite{Wallach02} showed that for a multiqudit system in which each subsystem has dimension at least three, Gleason's theorem can be proved using only rank-one unentangled projections. Specifically, frame functions on rank-one unentangled measurements can always be described by the Born rule, as in Eq.~\eqref{bornrule}. However, if even one subsystem has dimension two, the result fails to hold, i.e.~there exist frame functions on unentangled projections which cannot be described as in Eq.~\eqref{bornrule}.

To address the case of the KS theorem, we will be interested in rays in $\C^2$, or \emph{qubit rays}, and rays in ${\C^2}^{\otimes n}$, or \emph{n-qubit rays}. Given a pair of orthogonal unit vectors $\{\ket{0},\ket{1}\}$ in $\C^2$ a generic qubit ray can be expressed as 
\begin{equation}
\ket{\psi}=\cos(\theta/2)\ket{0}+e^{i\phi}\sin(\theta/2)\ket{1}\,,
\end{equation} 
for some $0\leq\theta\leq\pi$ and $0\leq\phi<2\pi$. The parameters $\theta$ and $\phi$ define a point on the Bloch sphere (see Fig.~\ref{hemi}). The pairs of orthogonal rays in $\C^2$ are given exactly by the pairs of antipodal points of the Bloch sphere. The product rays in ${\C^2}^{\otimes n}$ are then rays admitting an expression $\ket{\psi}=\ket{\psi_1}\otimes\cdots\otimes\ket{\psi_n}$ where $\ket{\psi_j}$ are qubit rays for all $1\leq j\leq n$.

%
%

\section{Kochen--Specker theorem  for multiqubit systems}\label{sec:mainresult}
In this section, we will show that any logical proof of the KS theorem on a multiqubit system requires entangled measurements, i.e.
\begin{restatable}{thm}{main}\label{thm:main}
	Any multiqubit Kochen--Specker set necessarily contains entangled rays.
\end{restatable}

We consider a KS-colouring $c_n$ on product rays of ${\C^2}^{\otimes n}$ defined in terms of a KS-colouring $c_1$ on the rays of $\C^2$. Specifically, $c_n(\ket{\psi_1}\otimes\cdots\otimes\ket{\psi_n})=\prod_{j=1}^nc_1(\ket{\psi_j})$. We choose $c_1$ to be particularly easy to visualise KS-colouring but the argument follows for any choice. To specify our KS-colouring, we first define a {\em north} qubit ray:
\begin{defn}
A \emph{north} qubit ray is of the form
\begin{equation}\ket{\psi}=\cos(\theta/2)\ket{0}+e^{i\phi}\sin{\theta/2}\ket{1}\end{equation} 
where $0\leq\theta<\pi/2$ and $0\leq\phi<2\pi$ or $\theta=\pi/2$ and $\pi<\phi\leq 2\pi$. We denote the set of north rays, depicted on the Bloch sphere in Fig.~\ref{hemi}, by $\mathcal{N}$. 
\end{defn}
\begin{figure}
\centering
\includegraphics[width=\columnwidth]{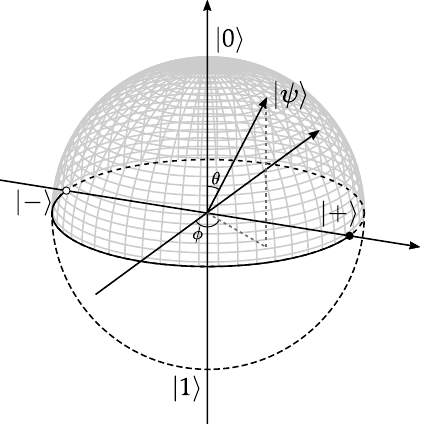}
\caption{\label{hemi}The Bloch sphere with the qubit north rays depicted by the hatched grey hemisphere including the solid point at $\ket{+}$ and the solid black half of the equator but not the open point at $\ket{-}$ or the dashed half of the equator.}
\end{figure}
Note that, for any north qubit ray $\ket{\psi}$, the ray $\ket{\psi^\perp}$ satisfying $\braket{\psi|\psi^\perp}=0$ is not a north ray, and vice versa, i.e.~every pair of orthogonal qubit rays contains exactly one north ray. This definition naturally extends to the case of $n$-qubit rays:
\begin{defn}
An \emph{all-north} $n$-qubit ray is a product ray $\ket{\psi}=\ket{\psi_1}\cdots\ket{\psi_n}$ such that $\ket{\psi_j}\in\mathcal{N}$ for all $j\in\left\{1,\ldots,n\right\}$. We denote the set of all-north $n$-qubit rays by $\mathcal{N}^n$.
\end{defn}
We need the following two lemmas in order to prove our first main result:
\begin{lem}\label{lem:2q}
Any complete set of two-qubit product rays contains exactly one all-north ray.
\end{lem}
\begin{proof}
A generic two-qubit product basis takes the form
\begin{equation}\label{eq:2basis}
\left\{\ket{\psi_1}\ket{\psi_2},\ket{\psi_1}\ket{\psi_2^\perp},\ket{\psi_1^\perp}\ket{\psi_3},\ket{\psi_1^\perp}\ket{\psi_3^\perp}\right\},
\end{equation}
for $\ket{\psi_1},\ket{\psi_2},\ket{\psi_3}\in\C^2$, noting that the systems may be swapped. By inspection, one can see that exactly one of the rays in Eq.~\eqref{eq:2basis} is all-north.
\end{proof}
The proof of the following lemma is based on that of Prop.~1 in \cite{CD17}.
\begin{lem}\label{lem:north}
Any complete set of $n$-qubit product rays contains exactly one all-north ray.
\end{lem}
\begin{proof}

We will prove this statement by induction on the number of qubits.

Firstly, an $n$-qubit complete set can contain at most one all-north ray since no two all-north rays are orthogonal.

Now, assume the result holds for $m$ qubits and consider a complete set of $(m+1)$-qubit product rays
\begin{equation}
P=\left\{\ket{\Psi_1}\ket{\psi_1},\ldots,\ket{\Psi_{2^{m+1}}}\ket{\psi_{2^{m+1}}}\right\},
\end{equation}
where $\ket{\Psi_j}\in{\C^2}^{\otimes m}$ and $\ket{\psi_j}\in\C^2$ for $j\in\left\{1,\dots,2^{m+1}\right\}$. Let $B$ denote the set of distinct rays of the final qubit, i.e.
\begin{multline}
B=\Big\{\ket{\psi}\in\C^2\Big|\ket{\Psi}\ket{\psi}\in P\\\text{ for some }\ket{\Psi}\in{\C^2}^{\otimes m}\Big\},
\end{multline}
noting that we may have $\ket{\psi_j}=\ket{\psi_k}$ for $j\neq k$, and let $E$
be a maximal set of non-orthogonal rays from $B$, i.e.~each ray in $B\backslash E$ is orthogonal to some ray in $E$. Denote by $J$ the subset of $\left\{1,\ldots,2^{m+1}\right\}$ such that $\ket{\psi_j}\in E$ for all $j\in J$.

Define the map $\mu$ to count the number of rays in the $(m+1)$-qubit basis, $P$, for which the final qubit is associated to a given ray, i.e.
\begin{equation}
\mu(\ket{\psi})=\abs{\left\{\ket{\Psi}\ket{\psi}\middle|\ket{\Psi}\in{\C^2}^{\otimes m},\ket{\Psi}\ket{\psi}\in P\right\}}.
\end{equation}
We may assume (without loss of generality) that the set $E$ is chosen such that 
\begin{equation}\label{eq:assum}
\mu(\ket{\psi})\geq\mu(\ket{\psi^\perp}),
\end{equation} for all $\ket{\psi}\in E$, since if $\mu(\ket{\psi})<\mu(\ket{\psi^\perp})$ then $\ket{\psi}$ could be replaced by $\ket{\psi^\perp}$ in $E$ without violating the requirements on the set $E$.

The maximality of $E$ implies that for every $\ket{\psi}\in B\backslash E$ we have $\ket{\psi^\perp}\in E$. Let $\mu(E)=\sum_{\ket{\psi}\in E}\mu(\ket{\psi})=\abs{J}$ and $\mu^\perp(E)=\sum_{\ket{\psi}\in E}\mu(\ket{\psi^\perp})$. Note that 
\begin{equation}
\mu(E)+\mu^\perp(E)=\sum_{\ket{\psi}\in B}\mu(\ket{\psi})=2^{m+1}
\end{equation}
and $\mu(E)\geq\mu^\perp(E)$ by the assumption in Eq. \eqref{eq:assum}.

If $\ket{\psi_j},\ket{\psi_k}\in E$ for $j\neq k$ then $\braket{\Psi_j|\Psi_k}=0$, since $\braket{\psi_j|\psi_k}\neq0$. Hence the set 
\begin{equation}\label{eq:mbasis}
\left\{\ket{\Psi_j}|j\in J\right\},
\end{equation}
comprises $\mu(E)$ mutually orthogonal rays of ${\C^2}^{\otimes m}$. It follows that $\mu(E)\leq 2^m$ which gives 
\begin{equation}
2\mu(E)\leq 2^{m+1}=\mu(E)+\mu^\perp(E),
\end{equation}
and $\mu(E)\leq\mu^\perp(E)$. Therefore, we find 
\begin{equation}\label{eq:mus}
\mu(E)=\mu^\perp(E).
\end{equation} Furthermore, it then follows from Eq. \eqref{eq:assum} that 
\begin{equation}\label{eq:muequal}
\mu(\ket{\psi})=\mu(\ket{\psi^\perp}),
\end{equation} for all $\ket{\psi}\in E$ and hence for all $\ket{\psi}\in B$ (since, by the maximality of $E$, if $\ket{\psi}\in B$ is not contained in $E$ then $\ket{\psi^\perp}\in E$). The definition of $\mu$ is independent of $E$, and we find that assumption \eqref{eq:assum} holds with equality for any choice of $E$. We may now choose $E$ to consist entirely of north rays: given any maximal set $E'$ of non-orthogonal rays from $B$ we have by Eq. \eqref{eq:muequal} that $\ket{\psi^\perp}$ is in $B$ for any $\ket{\psi}\in E'$ and therefore all $\ket{\psi}\in E'$ that are not north rays can be exchanged with the north rays $\ket{\psi^\perp}$ to produce the all-north set $E$.

We also find from Eq. \eqref{eq:mus} that $\mu(E)=2^m=|J|$  so that the set of Eq.~\eqref{eq:mbasis} is a complete set of product rays in ${\C^2}^{\otimes m}$. By our assumption (for the proof by induction), this $m$-qubit complete set contains exactly one all-north ray, say $\ket{\Psi_k}$. The ray $\ket{\Psi_k}\ket{\psi_k}\in P$ is then an all-north ray since $k\in J$ implies that  $\ket{\psi_k}\in E$, which is an all-north set of rays. Since $P$ can contain at most one all-north ray, it contains exactly one all-north ray.

By Lemma \ref{lem:2q}, since the statement of Lemma \ref{lem:north} holds for $n=2$, it holds for all $n\in\mathbb{N}$ by induction.
\end{proof}

We can now prove Theorem \ref{thm:main}.

\main*
\begin{proof}
We will show the contrapositive: the contextuality scenario generated by all $n$-qubit product rays is KS-colourable.
Define $c_n:\Sigma\left({\C^2}^{\otimes n}\right)\rightarrow\left\{0,1\right\}$ as follows.
\begin{equation}
c_n\left(\ket{\psi}\right)=
\begin{cases}
1\text{ if }\ket{\psi}\in\mathcal{N}^n,\\
0\text{ otherwise.}
\end{cases}
\end{equation}
By Lemma \ref{lem:north}, every complete set of $n$-qubit product rays contains exactly one all-north ray, hence  $c_n$ defines a KS-colouring on the contextuality scenario generated by all $n$-qubit product rays.
\end{proof}

\section{A Kochen--Specker noncontextual ontological model for unentangled n-qubit systems}\label{sec:model}
In the previous section we demonstrated the existence of KS-colourings of the unentangled rays of $n$-qubit systems. Such colourings define ontic states in a KS-noncontextual ontological model. However, the existence of such ontic states alone is insufficient to conclude that some quantum statistics can be reproduced by such a model.

In this section we will show how these ontic states can, indeed, yield a KS-noncontextual ontological model for the  fragment of quantum theory consisting of $n$-qubit product states and $n$-qubit unentangled measurements. We do so via an extension of the single qubit model of Kochen and Specker \cite{KS67} using the formulation of Leifer \cite{Leifer14}.

Before we proceed further, we note that if we were to restrict our analysis to measurements that can be performed locally on each qubit---either only measurements in \emph{direct} product bases (see Eq.~\eqref{eq:direct} or Ref.~\cite{mcnulty2016mutually}) or, more generally, those that can be implemented via LOCC---construction of a KS-noncontextual model would not require the results from the previous section. However, we are considering the strictly larger class of all unentangled measurements (see Sec.~\ref{sec:prelims}) including those that cannot be implemented via LOCC, for example, a measurement in the basis of  Eq.~\eqref{eq:nlbasis}. Thus, our KS-noncontextual ontological model requires KS-colourings of \textit{all} unentangled multiqubit measurement bases and not only those implementable via LOCC.
With these clarifications out of the way, we can now proceed to detail our construction.

Given a unit vector $\ket{\psi}$ in Hilbert space we will denote by $[\psi]$ the rank-one projection onto the subspace spanned by $\ket{\psi}$. The ontic state space for the model is given by $\Lambda=S_1^2\times\cdots\times S_n^2$, where $S_j^2$ is a two-dimensional sphere representing the ontic state space of the $j$th qubit ($j=1,2,\dots,n$) and $\times$ is the Cartesian product. We represent each $n$-qubit ontic state $\lambda=(\lambda_1, \lambda_2, \dots, \lambda_n)\in\Lambda$ by a single vector $\vec{\lambda}$ in a real vector space $\mathbb{R}^{3n}$:
\begin{align}
\vec{\lambda}=&\vec{\lambda}_1+\vec{\lambda}_2+\cdots\nonumber\\
=&\begin{pmatrix}\sin\theta_1\cos\varphi_1\\\sin\theta_1\sin\varphi_1\\ \cos\theta_1\\0\\0\\0\\0\\\vdots\\0 \end{pmatrix}+\begin{pmatrix}0\\0\\0\\\sin\theta_2\cos\varphi_2\\\sin\theta_2\sin\varphi_2\\ \cos\theta_2\\0\\\vdots\\0 \end{pmatrix}+\cdots
\end{align} 
where $0\leq\theta_j \leq \pi$ and $-\pi<\varphi_j\leq\pi$ are the coordinates of $\lambda$ in the $j$th copy of $S^2$.
Similarly, we may describe an unentangled rank-one projection $[\psi]=[\psi_1\otimes\cdots\otimes\psi_n]$ by a vector
\begin{align}\label{eq:vectorpsi}
\vec{\psi}=&\vec{\psi}_1+\vec{\psi}_2+\cdots\nonumber\\
=&\begin{pmatrix}\sin\theta^\psi_1\cos\varphi^\psi_1\\\sin\theta^\psi_1\sin\varphi^\psi_1\\ \cos\theta^\psi_1\\0\\0\\0\\0\\\vdots\\0 \end{pmatrix}+\begin{pmatrix}0\\0\\0\\\sin\theta^\psi_2\cos\varphi^\psi_2\\\sin\theta^\psi_2\sin\varphi^\psi_2\\ \cos\theta^\psi_2\\0\\\vdots\\0 \end{pmatrix}+\cdots
\end{align}
where $0\leq\theta^\psi_j\leq\pi$ and $-\pi<\varphi^\psi_j\leq\pi$ are the spherical coordinates of $\vec{\psi_j}$ on the $j$th sphere $S^2$.

We will now define a deterministic response function whereby, essentially, an outcome $[\psi]$ occurs with probability one in a state $\lambda$ exactly when $\vec{\psi_j}$ is in the hemisphere centred around $\vec{\lambda}_j$ for all $j\in\{1,2,\dots,n\}$. First, we make the following definitions. Let $\ket{\lambda_j}$ denote the Hilbert space vector with spherical coordinates $(\theta_j,\varphi_j)$ on the $j$th qubit, i.e.~$\ket{\lambda_j}=\cos(\theta_j/2) \ket{0} + \mathrm{e}^{i\varphi_j} \sin(\theta_j/2)\ket{1}$. Then we define a unitary that maps $\ket{\lambda_j}$ to the north pole $\ket{0}$, i.e.~ $U_{{\lambda}_j}=\ket{0}\bra{\lambda_j}+\ket{1}\bra{\lambda_j^\perp}$.

Now, explicitly, given an unentangled projection $[\psi]=[\psi_1\otimes\cdots\otimes\psi_n]$, the probability of observing an outcome associated with $[\psi]$ of a measurement $M$ on a system in ontic state $\lambda$ is given by 
\begin{equation}\label{response}
	\Prob\left([\psi]|M,\lambda\right)=\begin{cases}1&\text{ if } \bigotimes_{j=1}^nU_{\lambda_j}\ket{\psi}\in\mathcal{N}^n,\\0&\text{ otherwise }\,.\end{cases}
\end{equation}
Lemma \ref{lem:north} shows that Eq.~\eqref{response} defines a valid response function for $n$-qubit product projections.

The following inequalities will be useful in proving that our ontological model reproduces quantum theory:
\begin{equation}\label{heavbounds}
	\prod_{j=1}^nH^0(\vec{\psi}_j\cdot\vec{\lambda})\leq\Prob\left([\psi]|M,\lambda\right)\leq\prod_{j=1}^nH^1(\vec{\psi}_j\cdot\vec{\lambda})
\end{equation}
for all unentangled $n$-qubit projections $[\psi]$ and $\lambda\in\Lambda$, where $H^{0,1}:\mathbb{R}\rightarrow\{0,1\}$ are two conventions for the Heaviside step function, 
\begin{equation}
	H^y(x)=\begin{cases}1 & \text{if }x>0\\
		y& \text{ if } x=0\\
		0 & \text{ if } x< 0\,.\end{cases}
\end{equation}

We now define the epistemic states of the model. Given that a product state $[\chi]=[\chi_1\otimes\chi_2\otimes\cdots\otimes\chi_n]$ (which can be described by a vector as in Eq. \eqref{eq:vectorpsi} with parameters $0\leq\theta^\chi_j\leq\pi$ and $-\pi<\varphi^\chi_j\leq\pi$) is prepared, the probability measure over the ontic states of the system is given by
\begin{equation}
	\mu_\chi(\Omega)=\int_\Omega \prod_{j=1}^n p_{\vec{\chi}_j}(\vec{\lambda})\sin \theta_jd\theta_jd\varphi_j,
\end{equation}
for $\Omega\subseteq\Lambda$, where 
\begin{equation}
	p_{\vec{\chi}_j}(\vec{\lambda})=\frac{1}{\pi}H^0(\vec{\chi}_j\cdot\vec{\lambda})\vec{\chi}_j\cdot\vec{\lambda}\,.
\end{equation}

Thus the probability of observing outcome $[\psi]$ of a measurement $M$ performed on a system prepared in quantum state $[\chi]$ is given by 
\begin{align}
&\Prob\left([\psi]|M,[\chi]\right)\nonumber\\=&\int_\Lambda\Prob\left([\psi]|M,\lambda\right)d\mu_\chi(\lambda)\nonumber\\
=&\int_\Lambda\Prob\left([\psi]|M,\lambda\right) \prod_{j=1}^np_{\vec{\chi}_j}(\vec{\lambda})\sin \theta_jd\theta_jd\varphi_j
\end{align}
Since the terms of the integrand are all non-negative, it follows from Eq.~\eqref{heavbounds} that 
\begin{align}
	&\int_\Lambda\prod_{j=1}^nH^0(\vec{\psi}_j\cdot\vec{\lambda}) \prod_{j=1}^np_{\vec{\chi}_j}(\vec{\lambda})\sin \theta_jd\theta_jd\varphi_j\\
	\leq&\Prob\left([\psi]|M,[\chi]\right)\\
	\leq&\int_\Lambda\prod_{j=1}^nH^1(\vec{\psi}_j\cdot\vec{\lambda}) \prod_{j=1}^np_{\vec{\chi}_j}(\vec{\lambda})\sin \theta_jd\theta_jd\varphi_j
\end{align} 

For both of these bounds ($y\in\{0,1\}$) we find
\begin{align}
	&\int_\Lambda\prod_{j=1}^nH^y(\vec{\psi}_j\cdot\vec{\lambda}) \prod_{j=1}^np_{\vec{\chi}_j}(\vec{\lambda})\sin \theta_jd\theta_jd\varphi_j\\
	=&\int_\Lambda\prod_{j=1}^n\frac{1}{\pi}H^y(\vec{\psi}_j\cdot\vec{\lambda})H^0(\vec{\chi}_j\cdot\vec{\lambda})\vec{\chi}_j\cdot\vec{\lambda}\nonumber\\
	&\sin \theta_jd\theta_jd\varphi_j\\
	\label{line:product1}=&\prod_{j=1}^n\int_0^\pi\int_{-\pi}^\pi\frac{1}{\pi}H^y(\vec{\psi}_j\cdot\vec{\lambda})H^0(\vec{\chi}_j\cdot\vec{\lambda})\vec{\chi}_j\cdot\vec{\lambda}\nonumber\\&\sin \theta_jd\theta_jd\varphi_j.
\end{align}

We can now evaluate each term in this product for both cases $y=0\text{ and }1$, following the method of Ref.~\cite{Leifer14} (where the $y=0$ case is considered), to find 
\begin{align}\label{finalint}
&\int_0^\pi\int_{-\pi}^\pi\frac{1}{\pi}H^y(\vec{\psi}_j\cdot\vec{\lambda})H^0(\vec{\chi}_j\cdot\vec{\lambda})\vec{\chi}_j\cdot\vec{\lambda}\nonumber\\
&\sin \theta_jd\theta_jd\varphi_j\nonumber\\
&=|\langle\psi_j,\chi_j\rangle|^2,
\end{align}
for both $y=0,1$. The argument goes as follows. Firstly, we choose our coordinates $\theta_j$ and $\varphi_j$ such that $\vec{\chi}_j=(1,0,0)$ and $\vec{\psi}_j=(\cos\phi,\sin\phi,0)$ for some $-\pi<\phi\leq\pi$. Then since $\vec{\lambda}_j=(\sin\theta_j\cos\varphi_j,\sin\theta_j\sin\varphi_j,\cos\theta_j)$, we find
\begin{align}
\vec{\chi}_j\cdot\vec{\lambda}_j&=\sin\theta_j\cos\varphi_j,\\
\vec{\psi}_j\cdot\vec{\lambda}_j&=\sin\theta_j\cos(\varphi_j-\phi).
\end{align}
The integrand of Eq. \eqref{finalint} is non-zero when $\vec{\chi}_j\cdot\vec{\lambda}_j$ is positive and when $\vec{\psi}_j\cdot\vec{\lambda}_j$ is positive (non-negative) for the case $y=0$ ($y=1$). These conditions are achieved for $y=0$ when $-\pi/2<\varphi_j<\pi/2$ and $-\pi/2+\phi<\varphi_j<\pi/2+\phi$ and similarly for $y=1$ but when the latter inequalities are no longer strict, i.e.~$-\pi/2+\phi\leq\varphi_j\leq\pi/2+\phi$. Since the integrals over a closed or open interval are equal, in both cases $y=0\text{, or }1$, if $\phi$ is non-negative we find
\begin{align}
&\begin{aligned}\int_0^\pi\int_{-\pi}^\pi\frac{1}{\pi}H^y(\vec{\psi}_j\cdot&\vec{\lambda})H^0(\vec{\chi}_j\cdot\vec{\lambda})\\&\vec{\chi}_j\cdot\vec{\lambda}\sin \theta_jd\theta_jd\varphi_j\nonumber
\end{aligned}
\\=&\frac1\pi\int_0^\pi\sin^2\theta_jd\theta_j\int^{\frac{\pi}{2}}_{-\frac{\pi}{2}+\phi}\cos\varphi_jd\varphi_j\\
=&\frac12\left(1+\cos\phi\right)=|\langle\psi_j,\chi_j\rangle|^2\,\nonumber.
\end{align}
We find the same value if $\phi$ is negative. Finally, we have shown
\begin{equation}
\begin{aligned}
	\Prob\left([\psi]|M,[\chi]\right)&=\prod_{j=1}^n|\langle\psi_j,\chi_j\rangle|^2\\&=|\langle\psi,\chi\rangle|^2\,.
\end{aligned}
\end{equation}
The model is easily extended to higher rank projections via $\Pr\left(\sum_j[\psi^j]|M,\lambda\right)=\sum_j\Pr\left([\psi^j]|M,\lambda\right)$ giving $\Pr\left(\sum_j[\psi^j]|M,[\chi]\right)=\sum_j\Pr\left([\psi^j]|M,[\chi]\right)$ for any mutually orthogonal projections $[\psi^j]$. The fact that this extension is well-defined follows from Lemma~\ref{lem:north}.
Thus the model reproduces the predictions of quantum theory for $n$-qubit product states and unentangled measurements.

Note that for the case $n=1$, this ontological model is equivalent to the Kochen--Specker model for projective measurements and pure states of a single qubit \cite{KS67, Leifer14}. Our generalization, however, is nontrivial because it crucially relies on Lemma~\ref{lem:north} to define the $n$-qubit response function of Eq.~\eqref{response}. For a single qubit, on the other hand, Lemma~\ref{lem:north} is trivial because every qubit state appears in exactly one basis which already implies that each basis contains exactly one north state.

The model can be extended to include mixtures of product states, e.g.~$\sum_{i=1}^Nq_i[\chi^i]$, for product states $\ket{\chi^i}$ by simply defining the probability measure for such a mixture as the same mixture of probability measures, i.e.~$\sum_{i=1}^Nq_i\mu_{\chi^i}$. However, note that such an extension is \emph{preparation contextual}, in the sense that there is no fixed probability measure for a given separable density operator since different decompositions of the same density operator as a mixture of product states result in different probability measures. Indeed, there exists no preparation noncontextual extension to mixed states, as follows from the impossibility of a preparation noncontextual ontological model for single qubit mixed states \cite{Spekkens05}.

An immediate consequence of the existence of this KS-noncontextual ontological model is that it allows us to obtain a tight statement on the relationship between entanglement and KS-contextuality in multiqubit systems. Namely, entanglement is necessary for any proof of the KS theorem, whether logical \cite{KS67, KS15} or statistical \cite{KCBS, KS18}. Theorem \ref{thm:main} showed that multiqubit entangled measurements are necessary for logical proofs of the KS theorem. Our construction of the KS-noncontextual ontological model implies that even for statistical proofs of the KS theorem on multiqubit systems, one requires entanglement, either in the state or in the measurements. A well-known instance of such a statistical proof is the violation of Bell inequalities by implementing local (hence, unentangled) measurements on an entangled multiqubit state. One can also construct statistical proofs of the KS theorem (where there exist KS-colourings of the measurements involved) in which the quantum state is a product state and there is entanglement in the measurements. For example, given a multiqubit entangled state that violates a Bell inequality, one can apply a global unitary to the state making it a product state and then apply the same unitary to all the measurements involved in the violation. This transformation preserves the probabilities, the orthogonality relations, and hence, the statistical proof, whilst inevitably making some of the measurements entangled.

\section{Bell and Kochen--Specker coincide for multiqubit systems with unentangled measurements}\label{sec:bellks}
In the previous section we found a KS-noncontextual ontological model for unentangled projective measurements and product pure states of multiqubit systems. The model straightforwardly extends to separable states, allbeit in a preparation contextual way. It follows that to obtain a multiqubit statistical proof of the Kochen--Specker theorem from unentangled measurements one must employ an entangled multiqubit state.\footnote{Note the contrast between preparation contextuality and KS-contextuality for multiqubit systems: the former requires no entanglement but for the latter, entanglement is necessary. It is known that any preparation noncontextual ontological model of a fragment of quantum theory is also necessarily KS-noncontextual but not conversely~\cite{LM13, Kunjwal16}. We can thus conclude that in any ontological model of multiqubit systems, KS-contextuality implies both, entanglement and preparation contextuality.} Is entanglement of the multiqubit state, however, sufficient to yield such a statistical proof? In this section, we find the answer to be negative, i.e.~there must exist an KS-noncontextual ontological model that includes---besides separable states and unentangled projective measurements---certain entangled states. Specifically, we find that the multiqubit states that can demonstrate KS-contextuality with unentangled measurements (in a finite contextuality scenario) are exactly those that  can violate a Bell inequality with local projective measurements.

The fact that there exist entangled multiqubit states yielding a statistical proof of the KS theorem with only unentangled measurements follows from the existence of entangled multiqubit states that violate Bell inequalities with only local projective measurements. This implication can be seen from both the more traditional \emph{observable-based} perspective \cite{Kunjwal15}, as well as the \emph{event-based} perspective \cite{AFL15,Gonda}. 

In the observable-based perspective we consider the local observables of each party in a Bell experiment. The set of local observables corresponding to one choice of setting for each party commute pairwise and, thus, form a context. We may then consider the KS-noncontextual ontological models respecting all such contexts. The mathematical constraints defining these models in this situation are equivalent to the constraints imposed by Bell's assumption of local causality in the original Bell scenario.\footnote{Both sets of constraints can be viewed as instances of the marginal problem for classical probability distributions via Fine's theorem~\cite{Fine82a, Fine82b, CF12, Abramsky2011, Kunjwal15}.} Hence, a quantum violation of a Bell inequality using projective measurements also provides a set of quantum statistics deriving from local, and hence unentangled, measurements that are incompatible with a KS-noncontextual ontological model.

To prove our result the event-based perspective will be more convenient. In order the see the connection between Bell scenarios and KS-contextuality in this perspective one needs to go beyond the idea of KS-colourings of contextuality scenarios to that of \emph{probabilistic models} on contextuality scenarios following the framework of Ref.~\cite{AFL15}.

In the following subsection, we define probabilistic models on contextuality scenarios and related notions, then show how a proof of Bell's theorem yields a statistical proof of the KS theorem (a known connection). In the next subsection, we then show that the converse relationship also holds for multiqubit systems and unentangled measurements. Thus, we arrive at the second main contribution of this work: a multiqubit entangled state can yield a statistical proof of the KS theorem with unentangled measurements if and only if it violates a Bell inequality with local projective measurements.

\subsection{Bell implies KS}
A probabilistic model is a probability assignment to the vertices of a contextuality scenario (i.e.~a hypergraph) such that the probabilities assigned in each hyperedge sum to one. Explicitly, given a contextuality scenario $H$ with vertices $V(H)$ and hyperedges $E(H)$, a probabilistic model on the contextuality scenario is a map $p:V(H)\rightarrow[0,1]$ such that $\sum_{v\in e}p(v)=1$ for all $e\in E(H)$. A KS-colouring is a probabilistic model that only takes values $0$ and $1$. A \emph{classical model} is a probabilistic model that can be decomposed into a convex combination of KS-colourings, where a convex combination, $q$, of probabilistic models $p$ and $p'$ is the probabilistic model $q(v)=\omega p(v)+(1-\omega)p'(v)$ for some $\omega\in[0,1]$ and all $v\in V(H)$. 

A probabilistic model $p$ is quantum if and only if for some separable Hilbert space $\cH$ there exists (i) a projection $\Pi_v$ for every $v\in V(H)$ such that $\sum_{v\in e}\Pi_v=\I_\cH$ for all $e\in E(H)$, and (ii) a density operator $\rho$ on $\cH$, such that $p(v)=\Tr(\Pi_v\rho)$. 

A quantum model that is not classical is said to provide a {\em statistical proof} of the KS theorem \cite{KCBS, KS18}. Such nonclassicality can be witnessed by the violation of a Bell-KS inequality that bounds the polytope of classical models, e.g.~the KCBS inequality \cite{KCBS}. On the other hand, if a contextuality scenario admits no classical models (hence no KS-colourings) but it admits a quantum model, we have a {\em logical proof} of the KS theorem \cite{KS67, KS15} which demonstrates a stronger form of KS-contextuality: not only does there exist a quantum model outside the set of classical models, the set of classical models is, in fact, empty, i.e.~no KS inequalities exist and {\em every} probabilistic model (hence every quantum model) on the contextuality scenario fails to be  classical.

A Bell scenario is an experiment in which $n$ parties each perform a measurement $x_r$ and observe an output $a_r$ for $1\leq r\leq n$ on some system such that the parties' individual experiments are space-like separated. We call a measurement on the entire system performed by all $n$ parties a global measurement. Each global measurement is specified by a list of the local measurement settings $\X=x_1x_2\ldots x_n$ and has a set of possible outcomes $\A|\X$, where $\A=a_1a_2\ldots a_n$. We consider the case in which each local measurement has two possible outcomes, i.e.~$a_r\in\{0,1\}$ for all $1\leq r\leq n$.  This is the natural setting for a multiqubit Bell experiment where the parties implement local projective measurements.

A \emph{behaviour} in a Bell scenario is a probability distribution, $p(\A|\X)$, on the measurement outcomes. A local behaviour is a behaviour that can be decomposed into a convex combination of local \emph{deterministic} behaviours,
i.e.~of behaviours $p(\A|\X)=p_1(a_1|x_1)\cdots p_n(a_n|x_n)$, where $p_r(a_r|x_r)\in\{0,1\}$ denotes the probability of party $r$ observing $a_r$ given they performed measurement $x_r$ for all $1\leq r\leq n$.

A behaviour that can be realised by each party performing a projective measurement on some subsystem of a quantum system is called a projective quantum behaviour. For a projective quantum behaviour, $p(\A|\X)$, there exists a separable Hilbert space, $\cH=\cH_1\otimes\cdots\otimes\cH_n$, a projective measurement $\{\Pi^{r,x_r}_{a_r}\}_{a_r}$ on $\cH_r$ for each measurement setting $x_r$ of the $r$-th party for all $1\leq r\leq n$, and a density operator $\rho$ on $\cH$ such that $p(\A|\X)=\Tr(\bigotimes_{r=1}^n\Pi^{r,x_r}_{a_r}\rho)$.

A Bell scenario can be mapped to a contextuality scenario which has (i) a vertex for every possible global measurement outcome $\A|\X$ for all $\A$ and $\X$, (ii) a hyperedge consisting of all the possible outcomes $\A|\X$ of fixed global measurement $\X$, i.e.~each set $\left\{(\A|\X)\middle|a_r\in\{0,1\}\text{ for all }1\leq r \leq n\right\}$ is a hyperedge, \emph{and} (iii) a hyperedge deriving from each no-signalling condition. The non-signalling hyperedges consist of sets of possible outcomes of (hypothetical) \emph{adaptive} measurements in which one party performs a measurement first and depending on their outcome a second party selects a measurement setting and so on \cite{AFL15}. For example, the hyperedge $\{00|00,01|00,10|01,11|01\}$ features in the hypergraph of the two party Bell scenario with binary inputs and outputs. This hyperedge can be thought of as the outcomes of an adaptive measurement protocol in which the first party always chooses measurement setting $0$, and the second party chooses the measurement setting $0$ if the first party obtains outcome $0$ and $1$ if they obtain $1$.\footnote{Note that the global measurements from the Bell scenario are also examples of adaptive measurements in which each party's measurement setting does not actually change based on the outcome of any other party. Thus all the hyperedges in a contextuality scenario representing a Bell experiment can be thought of as adaptive measurements.}

There is a bijection between local behaviours in the Bell scenario and classical models in the corresponding contextuality scenario. Furthermore, each quantum behaviour deriving from projective measurements in the Bell scenario defines a quantum model in the contextuality scenario. A collection of multiqubit projective measurements and a density operator that violate a Bell inequality then generate a quantum model on the corresponding contextuality scenario that is not a classical model and thus yields a statistical proof of the KS theorem. For a detailed account, see Ref.~\cite{AFL15}.

\subsection{KS implies Bell}
We have seen that states that can violate Bell inequalities also yield statistical proofs of the KS theorem with unentangled measurements. We now show that the converse relationship also holds. Namely, for any collection of product multiqubit rays and an entangled state which yield a statistical proof of the KS theorem, the entangled state necessarily violates a Bell inequality.

The proof will proceed along the following lines. Consider a hypergraph $H$ with a non-classical quantum model given by some multiqubit product rays, $\mathcal{S}$, and a density operator $\rho$. We will expand the set of multiqubit rays $\cS$ to a set of rays $\cS'$ that correspond to all the outcomes of a quantum measurement strategy in some Bell scenario, $\mathcal{B}(\mathcal{S})$.  For example, if $\cS$ consisted of $\{\ket{00},\ket{+1},\ket{0+}\}$ it would be expanded to
\begin{equation}
\begin{aligned}
\cS'=\{&\ket{00},\ket{01},\ket{10},\ket{11},
\\&\ket{+0},\ket{+1},\ket{-0},\ket{-1},
\\&\ket{0+},\ket{0-},\ket{1+},\ket{1-},
\\&\ket{++},\ket{+-},\ket{-+},\ket{--}\},
\end{aligned}
\end{equation}
which contains all the rays corresponding to the outcomes of measurements in a two-party Bell scenario where each party has two possible measurement settings given by $\{\ket{0},\ket{1}\}$ and $\{\ket{+},\ket{-}\}$.
We can then extend the hypergraph, $H$, to the hypergraph $G\supseteq H$ generated by this extended set of rays and their orthogonality relations. We can also extend the quantum model on $H$ given by $\cS$ and $\rho$ to a quantum model on $G$ given by $\cS'$ and $\rho$; this quantum model on $G$ continues to be non-classical. Let $H'$ be the hypergraph corresponding to the Bell scenario $\mathcal{B}(\cS)$. The hypergraph $G$ will contain all the vertices and hyperedges of the hypergraph $H'$ but possibly with additional hyperedges deriving from non-local bases in $\cS'$ such as the basis in Eq. \eqref{eq:nlbasis}. The set $\cS'$ and state $\rho$ also give a quantum model on the Bell hypergraph $H'$, since every hyperedge of $H'$ is contained in $G$. Finally, we will show, following \cite{augusiak2012tight}, that the sets of classical models on $G$ and $H'$ are identical despite the difference in hyperedges. Thus, the non-classical model on $G$ is also non-classical on $H'$, meaning $\rho$ yields a non-local, projective quantum behaviour in $\mathcal{B}(\cS)$. This leads us to the following theorem:

\begin{restatable}{thm}{bellks}\label{thm:bellks}
	Any multiqubit density operator, $\rho$, that can yield a statistical proof of the KS theorem with a finite set of unentangled projective measurements  can violate a Bell inequality with local projective measurements.
\end{restatable}

A full proof of Theorem \ref{thm:bellks} is given in Appendix \ref{proof:bellks}.

It follows that it must be possible to extend our KS-noncontextual ontological model to include, besides separable states, multiqubit entangled states (e.g.~Werner states) that cannot yield Bell violations when subjected to  local projective measurements. We leave the construction of such a multiqubit KS-noncontextual ontological model as an open problem for future work.
\section{Do we need fully entangled bases?}\label{sec:fullyent}
Our initial question concerning the necessity of entanglement in any multiqubit logical proof of the KS theorem \cite{KS67} was motivated by the presence of entanglement in the Peres-Mermin magic square (Fig.~\ref{peresmermin}). Another curious feature of the Peres-Mermin argument is the appearance of bases that are not merely entangled but, in fact, {\em fully entangled}, i.e.~these bases contain no product projections. Is there, then, an even stronger requirement on the entanglement needed for a multiqubit KS theorem, i.e.~just as the presence of entangled projections in a KS set is generic and not particular to the Peres-Mermin case, must a multiqubit KS set always contain fully entangled bases? Here we rule out this possibility by explicitly constructing a two-qubit KS set that does not require fully entangled bases, i.e.~every basis contains at least one product ray. 

Our construction makes use of Peres' 33-ray proof of the KS theorem in $\C^3$ which, when represented as a contextuality scenario, actually requires 57 rays (since the proof makes use of the orthogonality of certain pairs of the 33 rays for which the final orthogonal ray is missing). We construct a KS set by embedding two copies of the 57 Peres rays in two different three dimensional subspaces of $\C^2\otimes\C^2$. 

Denote by $\ket{\phi_j}$, for integers $1\leq j\leq57$, the rays in the Peres 57-ray contextuality scenario. The rays can be carved up into 40 distinct orthonormal bases.

We will now embed Peres' three-dimensional rays in the three-dimensional subspace of the two-qubit Hilbert space orthogonal to the ray $\ket{00}$. Given a ray $\ket{v}=\alpha\ket{0}+\beta\ket{1}+\gamma\ket{2}\in \cR(\C^3)$, denote $\ket{v^{00}}=\alpha\ket{01}+\beta\ket{10}+\gamma\ket{11}\in \cR(\C^2\otimes\C^2)$. For each of the 40 bases $\left\{\ket{\phi_a},\ket{\phi_{b}},\ket{\phi_{c}}\right\}$ among the Peres rays, we add the basis 

\begin{equation}\label{basistype1}
	\left\{\ket{00},\ket{\phi_a^{00}},\ket{\phi_b^{00}},\ket{\phi_c^{00}}\right\},
\end{equation}
to our hypergraph.

We now perform the analogous operation in the subspace orthogonal to $\ket{01}$, but we first transform Peres' rays by a unitary matrix
\begin{equation}U=\frac{1}{3}\left(
	\begin{array}{ccc}
		 1+\sqrt{2} & 
		\tau_-&
		\tau_+\\
		\tau_+ & 
		1+\sqrt{2} &
		\tau_-\\
		\tau_- & 
		\tau_+ & 
		 1+\sqrt{2} \\
	\end{array}
	\right),
\end{equation}
where $\tau_\pm=\left(2-\sqrt{2}\pm\sqrt{6}\right)/2$, and denote the new set of rays $\ket{U\phi_j}$. Note that the orthogonality relations between the rays are preserved under this transformation. Explicitly, we add the basis 
\begin{equation}\label{basistype2}
	\left\{\ket{01},\ket{U\phi_a^{01}},\ket{U\phi_b^{01}},\ket{U\phi_c^{01}}\right\},
\end{equation}
to our hypergraph for every basis $\left\{\ket{\phi_a},\ket{\phi_{b}},\ket{\phi_{c}}\right\}$, where analogously $\ket{v^{01}}=\alpha\ket{00}+\beta\ket{10}+\gamma\ket{11}\in\cR(\C^2\otimes\C^2)$ for any $\ket{v}=\alpha\ket{0}+\beta\ket{1}+\gamma\ket{2}\in\cR(\C^3)$.

The unitary $U$ has been chosen such that $\braket{\phi^{00}_j|U\phi^{01}_k}\neq0$ for all entangled $\ket{\phi^{00}_j}$ and $\ket{U\phi^{01}_k}$. It follows that the entangled rays of the hypergraph cannot be formed into a basis, i.e.~there are no fully entangled bases in our construction. This fact can be verified by consulting the Mathematica notebook available at Ref.~\cite{nb}. 
%

The contextuality scenario can now be seen to be KS-uncolourable as follows. Any KS-colouring, $c$, of our hypergraph should assign zero to at least one of the rays $\ket{00}$ and $\ket{01}$, since they appear in a hyperedge---the hyperedge consisting of $\left\{\ket{00},\ket{01},\ket{10},\ket{11}\right\}$ which is one of the bases given by Eq.~\eqref{basistype1}. If $\ket{00}$ is assigned zero by the colouring, $c$, then the assignments to the rays $\ket{\phi^{00}_j}$ would constitute a KS-colouring of the Peres hypergraph, which does not exist. Similarly, if $\ket{01}$ is assigned zero in the colouring, $c$, then the assignments to the rays $\ket{U\phi^{01}_j}$ would also constitute a KS-colouring of the Peres hypergraph. Hence such a colouring, $c$, cannot exist.

\section{Unentangled Kochen--Specker sets}\label{sec:unentks}

Between the case of multiqubit systems (wherein each subsystem has dimension two) and the case of multiqudit systems, wherein each subsystem has dimension at least three, we have the possibility of multiqudit systems consisting of both qubits and higher-dimensional qudits. For these systems, Wallach showed that unentangled projections are still  insufficient to yield Gleason's theorem \cite{Wallach02}. Is it, however, possible to obtain the KS theorem with unentangled projections in this case, unlike the multiqubit case? In this section, we show that this is indeed possible. In fact, the presence of just one qutrit in an otherwise multiqubit system is enough to allow constructions of KS sets with unentangled projections. Our argument, not surprisingly, relies on the fact that a single qutrit admits KS sets on its own \cite{KS67}.
\begin{restatable}{thm}{unentd}\label{unentd3}
	There exists a KS set consisting entirely of product rays in any separable Hilbert space $\cH_1\otimes\cdots\otimes\cH_n$ where $\dim(\cH_j)\geq3$ for some $1\leq j \leq n$.
\end{restatable}
Here we give an idea of the proof with an example and give the general proof in Appendix~\ref{thmpf}. We will show that if there were no KS set consisting of product rays in $\cH_1\otimes\cdots\otimes\cH_n$, i.e.~if there existed a KS-colouring $c$ on the product rays in $\cH_1\otimes\cdots\otimes\cH_n$, then we would be able to use $c$ to define a KS-colouring $c'$ on any set of bases of $\cH_j$, contradicting the KS theorem. For example, assume $c$ is a KS-colouring of the product rays in $\C^2\otimes\C^3$. Now take any set of bases $\{B_1,\ldots,B_N\}$ in $\C^3$, for example the bases from the Peres 33-vector proof ($57$ vectors in $40$ bases), and take the element-wise tensor product of each basis $B_k$ with the basis $\{\ket{0},\ket{1}\}$ of $\C^2$. The resulting set of bases, $\left\{\{\ket{0},\ket{1}\}\otimes B_k\right\}_k$, of $\C^2\otimes\C^3$ would be KS-colourable by the definition of $c$. Now given a basis $B_k$ we define a map $c'$ on $B_k$ by 
\begin{equation}
c'(\ket{\psi})=\begin{cases} 1 &\text{if } c(\ket{0}\ket{\psi})=1 \text{ or } c(\ket{1}\ket{\psi})=1\\ 0 &\text{otherwise.}\end{cases}
\end{equation}
For example, given the basis $\{\ket{0},\ket{1},\ket{2}\}$, if $c(\ket{02})=1$ then $c'(\ket{2})=1$ and $c'(\ket{0})=c'(\ket{1})=0$. By this definition and the fact that $c$ assigns one to exactly one element of $\{\ket{0},\ket{1}\}\otimes B_k$ we find that, likewise, $c'$ assigns one to exactly one element of $B_k$. Then we see that $c'$ is well-defined across all the elements of the bases $B_k$ since it is defined independently from the basis $B_k$. The map $c'$ would therefore be a KS-colouring of the bases $\{B_1,\ldots,B_N\}$, contradicting the result of Peres.

Note that the proof of Theorem~\ref{unentd3} in Appendix \ref{thmpf} actually shows the stronger result that the hypergraph of product rays of $\cH_1\otimes\cdots\otimes\cH_n$ with only the hyperedges deriving from \emph{direct} product bases is also KS-uncolourable. A direct product basis \cite{mcnulty2016mutually} is a basis formed by taking the elementwise tensor product of a basis for each subsystem $\cH_j$, i.e.~a basis
\begin{multline}\label{eq:direct}
\{\ket{\psi^1_1}\ket{\psi^2_1}\cdots\ket{\psi^n_1},\ket{\psi^1_2}\ket{\psi^2_1}\cdots\ket{\psi^n_1},\\\ldots,\ket{\psi^1_{N_1}}\ket{\psi^2_{N_2}}\cdots\ket{\psi^n_{N_2}}\}
\end{multline} 
where $\{\ket{\psi^j_1},\ldots,\ket{\psi^j_{N_j}}\}$ is a basis for each $\cH_j$ for $1\leq j\leq n$.

The contrast between the \textit{impossibility} of Gleason's theorem and the \textit{possibility} of KS theorem with unentangled measurements in the case of composite systems that contain both qubits and higher-dimensional qudits may seem surprising at first glance. To gain some intuition for this contrast, consider two important facts:  first, very simply, Gleason's theorem implies the KS theorem, but not conversely. Second, in more depth, the existence of a KS theorem but the lack of a corresponding Gleason's theorem reflects the fact that the former is a \textit{no-go theorem} (ruling out certain types of probabilistic assignments) while the latter is a \textit{``go theorem"} (specifying the allowed probabilistic assignments).

Consider, for example, a qubit-qutrit system. A single qubit does not permit a proof of Gleason's theorem because the structure of projective measurements on a qubit allows for many more probabilistic assignments than are dictated by the Born rule. Composing the qubit with a qutrit and considering unentangled measurements on the pair does not rule out these non-quantum assignments on the qubit, so Gleason's theorem continues to fail. On the other hand, a single qutrit admits a proof of the KS theorem and this no-go statement goes through even if we compose it with a qubit.

\section{Implications for the role of contextuality in quantum computation}\label{sec:qcomp}
The role of contextuality within the paradigm of quantum computation with state injection (QCSI) has been a subject of active research in recent years \cite{HWV14, BDB17}.\footnote{Although we focus on KS-contextuality in this paper, recent results \cite{SDS21} have shown connections between QCSI and generalised contextuality \cite{Spekkens05}, analogous to the conclusions of Ref.~\cite{HWV14} with respect to KS-contextuality.} This approach to quantum computation relies on lifting stabiliser quantum circuits, which cannot implement universal quantum computation, to universality via the injection of non-stabiliser states called {\em magic states}. It has been shown that, in the case of odd-prime dimensional qudit (or, {\em quopit} \cite{BDB17}) circuits, contextuality is a necessary resource for quantum computation, i.e.~the magic states needed for universality must exhibit contextuality with respect to stabiliser measurements \cite{HWV14}. Thus, if a state admits a KS-noncontextual ontological model with the stabiliser measurements, it cannot promote the circuit to universality. The converse claim, that (KS-)contextuality is sufficient for universal quantum computation, is conjectured but unproven \cite{HWV14}.

The multiqubit case (i.e.~the even-prime dimensional case) has, however, been a hurdle in interpreting contextuality of magic states as a resource for quantum computation \cite{HWV14}. The multiqubit stabiliser subtheory is classically efficiently simulable \cite{GK98, AG04} despite the presence of (state-independent) KS-contextuality in its observables (e.g.~see Fig.~\ref{peresmermin}). The fact that {\em any} state, stabiliser or not, exhibits contextuality with respect to such observables means that there is nothing special about the contextuality of magic states that renders contextuality a necessary resource for universal quantum computation. While magic states are still a necessary resource, their contextuality does not single them out (unlike the quopit case) since {\em all} states can exhibit contextuality.\footnote{This also means that, unlike the quopit case, contextuality cannot even be {\em conjectured} as a sufficient condition for universality in this case.} To overcome this hurdle, {\em restricted} QCSI schemes have been proposed \cite{BDB17} which restore the status of contextuality as a resource for quantum computation.

In Ref.~\cite{BDB17}, such a restricted QCSI scheme $\mathcal{M}_{\mathcal{O}}$ is required to satisfy:

\begin{enumerate}
	\item[(C1)] \emph{Resource character.} There exists a quantum state that
	does not exhibit contextuality with respect to measurements available in $\mathcal{M}_{\mathcal{O}}$.
	\item[(C2)] \emph{Tomographic completeness.} For any
state $\rho$, the expectation value of any Pauli observable can be inferred via the allowed operations of the scheme.
\end{enumerate}

The requirement (C1) means that the measurements in the scheme cannot exhibit state-independent contextuality. It follows from Theorem~\ref{thm:main} of the present manuscript that a sufficient condition for satisfying requirement (C1) is that every measurement in the scheme be unentangled. Specifically, the scheme $\mathcal{M}_{\mathcal{O}}$ prescribes the sets of observables in the scheme that are jointly measurable. If none of these joint measurements requires a measurement in an entangled basis then (C1) is satisfied.

The schemes proposed in Ref.~\cite{BDB17} satisfy exactly this condition; they contain no entangled measurements. Thus, all the contextuality in these schemes derives from entanglement of the injected state. Furthermore, it follows from Theorem~\ref{thm:bellks} that in order for the injected state to promote such a scheme to universality it must be capable of violating a Bell inequality with some local projective measurements.

\section{Conclusions}\label{sec:disc}

In this work we have demonstrated the necessity of entanglement in proofs of the KS theorem for systems of multiple qubits, i.e., KS-contextuality necessitates not only incompatibility but also entanglement in multiqubit systems.

On the one hand, we showed KS sets for multiqubit systems must contain entangled rays/projections. It follows that any \emph{logical} proof of the KS theorem for multiqubit systems relies upon entanglement in the measurements, just like in the case of the Peres-Mermin square (Fig.~\ref{peresmermin}). However, unlike the Peres-Mermin square, multiqubit proofs of KS-contextuality do not, in general, require measurements in {\em fully entangled} bases.

On the other hand, the KS-noncontextual ontological model defined in Sec.~\ref{sec:model} allows us to go further and also make statements about \emph{statistical} (and state-dependent) proofs of the KS theorem, i.e.~those proofs that are in the style of Klyachko {\em et al.}~\cite{KCBS}. Specifically, this model can be easily extended to include unentangled projective measurements and separable states of a multiqubit system. Thus, any statistical proof of the KS theorem for such a system must employ entanglement, either in the state or in the measurements (or both). For example, proofs of Bell's theorem give rise to statistical proofs of the KS theorem in which the measurements are unentangled but the state is necessarily entangled  \cite{CSW14, Kunjwal19}.

Our results also allow us to make the following comparisons with other forms of nonclassicality. 
	
Firstly, it follows from our results that the nonclassicality present in unentangled measurements in the form of  `nonlocality without entanglement' \cite{BDF99} is insufficient to witness the (KS-)contextuality of multiqubit systems.

Secondly, Bell's theorem follows from unentangled measurements and requires an  entangled state (although entanglement is not sufficient \cite{Werner89}). Similarly, it follows from the model in Sec.~\ref{sec:model} that statistical proofs of the KS theorem employing unentangled measurements also require entangled states for multiqubit systems. Moreover, in Theorem~\ref{thm:bellks}, we have shown that the subset of entangled sets that violate a Bell inequality with local projective measurements are exactly those which can yield a finite statistical proof of the KS theorem, where by `finite' we mean that the proof uses a finite set of measurements.

Finally, our results clarify the connection between Gleason's theorem and the KS theorem with respect to entanglement. While both theorems hold in dimensions greater than two and fail to hold in dimension two, their behaviour with respect to entanglement differs in multiqudit systems. In any multiqudit system that contains both, a subsystem of dimension two and another of dimension at least three, the two theorems diverge.  Unentangled measurements are sufficient for a proof of the KS theorem but not Gleason's theorem for such multiqudit systems (Fig.~\ref{unentksgleason}).

We have also discussed the implications of our results for the program of understanding the role of contextuality in quantum computation. We find that the assumptions underlying some previously proposed QCSI schemes with qubits \cite{BDB17} become more intuitive in view of our results. 
	
The questions we have raised and addressed in this paper have implications for both fundamental and applied aspects of quantum theory. Further development of the applied aspects, particularly with respect to quantum computation, will be taken up in future work.

\section*{Acknowledgements}
We thank Tobias Fritz for invaluable discussions in the early stages of this work. This project has received funding from the European Union's Horizon 2020 research and
innovation programme under the Marie Sk{\l}odowska-Curie grant agreement No. 754510. VJW acknowledges support from the Government of Spain (FIS2020-TRANQI and Severo Ochoa CEX2019-000910-S), Fundaci\'o Cellex, Fundaci\'o Mir-Puig, Generalitat de Catalunya (CERCA, AGAUR SGR 1381 and QuantumCAT).  RK is supported by
the Charg\'e des recherches fellowship of the Fonds de la
Recherche Scientifique-FNRS (F.R.S.-FNRS), Belgium. This research was partly supported
by Perimeter Institute for Theoretical Physics through their visitor program. Research at Perimeter Institute is supported by the Government of Canada through the Department of Innovation, Science
and Economic Development Canada and by the Province of Ontario through the Ministry of Research, Innovation and Science.

\begin{appendix}
	\section{Traditional KS}\label{app:trad}
	
		A traditional proof of the KS theorem comprises a set of self-adjoint operators on a Hilbert space that lead to a contradiction when one attempts to find a \emph{valuation} (see Def.~\ref{def:val} below) on all such self-adjoint operators \cite{KS67,redhead1987incompleteness}. A subset of these operators that commute pairwise is known as a context, since all the observables in the subset are jointly measurable. In finite dimensions, and in particular for multiqubit systems, the operators in a context have a simultaneous orthonormal eigenbasis, i.e.~an orthonormal basis in which every vector is an eigenvector of every operator in the context. In this appendix we will describe the relationship between the existence of valuations and KS-colourings. Thus, we demonstrate how Theorem \ref{thm:main} implies that for multiqubit systems a set of operators leading to a traditional proof of the KS theorem must contain a context for which the simultaneous orthogonal eigenbasis contains entangled vectors. 
		
Let $\mathcal{L}_{\operatorname{sa}}(\mathcal{H})$ denote the self adjoint operators of a separable Hilbert space $\cH$.	
	\begin{defn}\label{def:val}
		A valuation is a map $v:\mathcal{L}_{\operatorname{sa}}(\mathcal{H})\rightarrow\mathbb{R}$ satisfying:
		\begin{itemize}[leftmargin=4em]
		\item[(SPEC)] $v(A)\in\sigma(A)$, where $\sigma(A)$ denotes the spectrum of $A$
		\item[(FUNC)] $v(g(A))=g(v(A))$ for any Borel function $g$.
		\end{itemize}
	\end{defn}
The requirement (FUNC) implies, for example, that for commuting operators $A$ and $B$ we have $v(AB)=v(A)v(B)$ and $v(aA+bB)=av(A)+bv(B)$ for $a,b\in\mathbb{R}$. Therefore, if there did exist a valuation, $v$, on $\mathcal{L}_{\operatorname{sa}}(\mathcal{H})$ it would assign values zero or one to each rank-one projection and satisfy $v(\Pi_1)+v(\Pi_2)+\ldots=1$ for sets of orthogonal projections $\{\Pi_1,\Pi_2,\ldots\}$ summing to the identity. In other words, the valuation would define a KS-colouring when restricted to the rank-one projections of the Hilbert space. The non-existence of such a colouring, stated in Theorem~\ref{thm:KS}, therefore implies the non-existence of a valuation. 
	
In the case of multiqubit product rays we have found that KS-colourings do, however, exist. 
We may use these colourings to define valuations on certain subsets of self-adjoint operators.

We define a traditional proof of the KS theorem as a set $\mathcal{A}$ of self-adjoint operators such that there does not exist a map $v$ on $\mathcal{A}$ satisfying (SPEC) and (FUNC), where (FUNC) is limited to Borel functions $g$ such that $g(A)\in\mathcal{A}$.\footnote{Note that under this definition the set of nine operators in the Peres-Mermin square does \emph{not} produce a KS contradiction, however, if one adds the identity and negative identity to the set then the standard contradiction can be found.} We find that such a set $\mathcal{A}$ must contain contexts for which the simultaneous orthogonal eigenbasis contains entangled vectors in multiqubit systems. 

	\begin{lem}\label{lem:val} Consider a set of self-adjoint operators $\mathcal{A}$ on ${\C^2}^{\otimes n}$ such that the simultaneous eigenbases for each context in $\mathcal{A}$ consist entirely of product vectors. Then $\mathcal{A}$ does not form a traditional proof of the KS theorem.
	\end{lem}
	\begin{proof}
By Theorem~\ref{thm:main}, there exists a KS-colouring, $c$, on the hypergraph $H$ generated by the simultaneous eigenbases of the contexts in $\mathcal{A}$. As in the treatment of Gleason's theorem in Sec.~\ref{sec:prelims}, this colouring can be equivalently defined on the rank-one projections corresponding to each ray and extended to higher rank projections via $c(\Pi_1+\Pi_2+\cdots)\equiv c(\Pi_1)+c(\Pi_2)+\cdots$  for mutually orthogonal sets of projections $\{\Pi_1,\Pi_2,\ldots\}$. Given $A\in\mathcal{A}$, let $A=\sum_{j}\lambda_j\Pi_j$ be its spectral decomposition, where $\Pi_j$ are orthogonal projections and $\lambda_j$ are the eigenvalues of $A$. We will show the map 
	\begin{equation}\label{eq:val}
			v(A)=\sum_{j}\lambda_jc\left(\Pi_j\right),
	\end{equation} satisfies SPEC and FUNC for all $A\in\mathcal{A}$.
		
For a given operator $A\in\mathcal{A}$, the colouring $c$ takes value one on exactly one of the projections $\Pi_j$ and zero on the rest. It follows that $v$ satisfies the (SPEC) principle. 
For any Borel function $g$ such that $g(A)\in\mathcal{A}$, we have
		\begin{align}
			v\left(g(A)\right)&=v\left(g\left(\sum_{j}\lambda_j\Pi_j\right)\right)\\
			&=v\left(\sum_{j}g(\lambda_j)\Pi_j\right)\\
			&=\sum_{j}g(\lambda_j)c(\Pi_j)=g(\lambda_k),
		\end{align}
		for some $\lambda_k\in\sigma(A)$. And further
		\begin{align}
			g(v(A))=g\left(\sum_{j=1}^d\lambda_jc(\Pi_j)\right)=g(\lambda_k).
		\end{align}
		Therefore, $v$ also satisfies the (FUNC) principle.
	\end{proof}
\section{Proof of Theorem \ref{thm:bellks}}\label{proof:bellks}

\bellks*

\begin{proof}
	Let 
	\begin{equation}
		\mathcal{S}=\left\{\ket{\psi_j}=\ket{\psi_j^1}\ket{\psi_j^2}\ldots\ket{\psi_j^n}|1\leq j\leq m\right\},
	\end{equation}
	be a set of $m$ product rays in $\C^{2^{\otimes n}}$, and let $\rho$ be a density operator on $\C^{2^{\otimes n}}$ that yields a statistical proof of the KS theorem with the hypergraph, $H$, generated by $\cS$ and all the bases contained in $\cS$. Explicitly, the hypergraph $H$ has $m$ vertices, $v_j$ for $1\leq j\leq m$, such that $\{v_j|j\in J\}\in E(H)$ if and only if $\{\ket{\psi_j}|j\in J\}$ is an orthonormal basis, where $J\subset\{1,\ldots,m\}$ is some indexing set. The non-classical quantum model on $H$ is given by $p^\cS_\rho(v_j)=\langle\psi_j|\rho|\psi_j\rangle$. We will show that $\rho$ necessarily violates a Bell inequality.
	
	Consider the set $\mathcal{S}_r$ of distinct local rays of the $r$-th system that are pairwise non-orthogonal. Explicitly, 
	\begin{equation}
		\mathcal{S}_r=\left\{\ket{\psi_j^r}|j\in J_r\right\},
	\end{equation} 
	where $j\in J_r\subseteq\left\{1,\ldots,m\right\}$ if and only if $\ket{\psi_j^r}\neq\ket{\psi_k^r}$ and $\braket{\psi_j^r|\psi_k^r}\neq 0$ for all $k<j$. Now we will relabel the states in $\mathcal{S}_r$ as follows 
	$$\mathcal{S}_r^{\rm re}=\left\{\ket{0^r_k}|1\leq k\leq \abs{\mathcal{S}_r}\right\}.$$
	
	Let $\mathcal{B}\left(\mathcal{S}\right)$ be the $n$ party Bell scenario in which the $r$-th party has $\abs{\mathcal{S}_r}$ binary measurement settings. Consider the quantum strategy in this scenario in which the $r$-th party's $k$-th measurement is given by 
	\begin{equation}\label{eq:rthmmts}
		x_r=k:\quad\left\{\ket{0^r_k}, \ket{1^r_k}\right\},
	\end{equation}
	where $\ket{1^r_k}\in\C^2$ is the ray orthogonal to $\ket{0^r_k}$. For $a_r\in\{0,1\}$, we denote the $n$-qubit ray $\bigotimes_{r=1}^n\ket{(a_r)^r_{x_r}}$ by $\ket{(\A|\X)}$---in correspondence with its outcome in the Bell scenario. Note that the global measurements, $\{\ket{(\A|\X)}\}_\A$, in this Bell experiment contain all the rays in $\mathcal{S}$. Let $\cS'$ denote the extended set (compared to $\cS$) of rays $\{\ket{(\A|\X)}\}_{\A,\X}$.
	
	Denote by $H'$ the contextuality scenario corresponding to the Bell scenario $\mathcal{B}(\cS)$. Recall that each hyperedge of $H'$ is given by the outcomes of an adaptive measurement. Equivalently, under the assignment of the rays from $\cS'$ to each vertex, each hyperedge corresponds to a projective measurement measurement that can be performed via LOCC. In other words, the hypergraph $H'$ is generated by $\cS'$ and those orthonormal bases that can be implemented via LOCC rather than all possible orthonormal bases. It follows that, although the vertices of $H'$ are a superset of those of the original contextuality scenario, $H$, there could be hyperedges in $H$ which are not contained in $E(H')$ since they arise from ``non-local'' bases in $H$, such as the basis in Eq.~\eqref{eq:nlbasis}. Denote by $G$ the contextuality scenario generated by rays $\cS'$ and \emph{all} the bases, which gives $V(G)=V(H')$ and $E(G)\supseteq E(H')\cup E(H)$. 
	
	The entangled state $\rho$ combined with the assignment of the ray $\ket{(\A|\X)}$ to each vertex $\A|\X$ gives a quantum model $p_\rho$ on both the hypergraphs, $G$ and $H'$. Further, we have that $p_\rho$ is a non-classical model on $G$ since it is a non-classical model on a subset of $G$, namely, $H$\footnote{The $G$ can only add further constraints on the classical models achievable on $H$, so any non-classical model on $H$ will necessarily be non-classical on $G$. Or, to take the contrapositive, if a model is classical on $G$ then it must be classical on $H$ because every deterministic model on $G$ will also be a deterministic model on $H$.}. Now, we will show that the classical models on $G$ and $H'$ coincide exactly, meaning that $p_\rho$ is also a non-classical model on $H'$. Given the bijection between classical models on $H'$ and local behaviours in the Bell scenario $\mathcal{B}(\cS)$, it follows that $\rho$ violates a Bell inequality.
	
	Clearly, a classical model on $G$ is a classical model on $H'$ since there are fewer constraints on the classical models on $H'$ due to the reduced set of hyperedges. We now show the converse, using an argument from \cite[Theorem~2]{augusiak2012tight}. Given an edge $e=\{\left(\mathbf{a}_j|\mathbf{x}_j\right)|1\leq j\leq 2^n\}$ of $G$ we have that for any pair of distinct vertices,  $\mathbf{a}_j|\mathbf{x}_j$ and $\mathbf{a}_k|\mathbf{x}_k$, in the edge the corresponding pair of rays, $\ket{(\mathbf{a}_j|\mathbf{x}_j)}$ and $\ket{(\mathbf{a}_k|\mathbf{x}_k)}$, are orthogonal by the definition of $G$. Since these rays are n-qubit product rays, we find that for at least one of the single-qubit subsystems, the qubit ray in $\ket{\mathbf{a}_j|\mathbf{x}_j}$ must the orthogonal to qubit ray of the same subsystem of $\ket{\mathbf{a}_k|\mathbf{x}_k}$. By construction, for each qubit ray that occurs in a given subsystem the orthogonal ray only occurs as the other outcome of the same measurement setting, cf.~Eq.~\eqref{eq:rthmmts}. If this orthogonality occurs in the $r$-th subsystem, explicitly, we find $(x_r)^j=(x_r)^k$ and $(a_r)^j\neq(a_r)^k$. Thus, we find in the two events $\mathbf{a}_j|\mathbf{x}_j$ and $\mathbf{a}_k|\mathbf{x}_k$ of the Bell scenario, party $r$ has the same measurement setting $(x_r)^j=(x_r)^k$ but observes a different outcome, either $(a_r)^j$ or $(a_r)^k$, i.e.~the two events are locally orthogonal in the terminology of Ref.~\cite{FSA13}.
	
	It follows that for any local deterministic behaviour in the Bell scenario (i.e.~classical model on $H'$), at most one of the events $\mathbf{a}_j|\mathbf{x}_j$ or $\mathbf{a}_k|\mathbf{x}_k$ occurs with probability one, whilst the other must occur with probability zero. Since this relationship holds between any pair of events in the edge $e\in E(G)\supseteq E(H')$, we find that a local deterministic behaviour on $H'$ assigns one to at most one of these outcomes. Hence, for each $e\in E(G)\backslash E(H')$, we have the following Bell inequality for $\mathcal{B}(\cS)$ (or, equivalently, a KS-noncontextuality inequality on the contextuality scenario $H'$)
	\begin{equation}\label{eq:ncineq}
		\sum_{\A|\X\in e}p_L(\A|\X)\leq 1,
	\end{equation}
	where $p_L$ is a local behaviour (or, equivalently, a classical model on $H'$).
	 	
	Finally, we show that the Bell inequality of Eq.~\eqref{eq:ncineq} is saturated by all local behaviours in $\mathcal{B}(\cS)$. We do so by showing that an internal point of the polytope of local behaviours saturates the inequality and therefore the inequality is trivial, in the sense that it is exactly saturated by all models in the affine span of the local polytope. Observe that the uniform behaviour $p_U\left(v\right)=1/2^n$ for all $v\in V(H')$ is an internal point of the local polytope and saturates the inequality \eqref{eq:ncineq}.
	
	Any internal behaviour $p_I(\A|\X)$ of the local polytope may be expressed as a convex combination 
	\begin{equation}
		p_I(\A|\X)=\omega p_D(\A|\X)+(1-\omega)p_\delta(\A|\X),
	\end{equation}
	of any local deterministic behaviour $p_D(\A|\X)$ (a vertex of the local polytope) and some other behaviour on the boundary of the local polytope $p_\delta(\A|\X)$, where $0\leq \omega\leq 1$. Thus, we have 
	\begin{equation}
		\begin{aligned}
			&\sum_{j=1}^{2^n}p_I(\A^j|\X^j)
			\\&=\omega\left(\sum_{j=1}^{2^n} p_D(\A^j|\X^j)\right)+(1-\omega)\left(\sum_{j=1}^{2^n}p_\delta(\A^j|\X^j)\right)\\&=1,
		\end{aligned}
	\end{equation}
	and, therefore, by the inequality \eqref{eq:ncineq} we find 
	\begin{equation}
		\sum_{j=1}^{2^n}p_D\left(\mathbf{a}^j|\mathbf{x}^j\right)=1.
	\end{equation} 
	Note that any affine combination of local deterministic behaviours, thus any non-signalling behaviour \cite[Corollary 2]{affineNS}, also obtains the value one for this Bell expression. In particular, we have
	\begin{equation}\label{eq:saturates}
		\sum_{j=1}^{2^n}p_L\left(\mathbf{a}^j|\mathbf{x}^j\right)=1,
	\end{equation}
	for any local behaviour $p_L$ in $\mathcal{B}(\cS)$.
	
	We have shown that any local behaviour in $\mathcal{B}(\cS)$, and thus classical model, $p_C$, on $H'$, satisfies $\sum_{v\in e}p_C(v)=1$ for all hyperedges $e\in E(G)$. Thus, the classical models on $H'$ are exactly the classical models on $G$. It follows that $p_\rho$ is a non-classical quantum model on $H'$ and therefore violates a Bell inequality.
\end{proof}
\section{Proof of Theorem~\ref{unentd3}}\label{thmpf}
\unentd*

\begin{proof}
	Let $\cH=\cH_1\otimes\ldots\otimes\cH_{n-1}$ be some separable Hilbert space for which there exists a KS set of bases $V^k$ for $1\leq k\leq K$ consisting entirely of product vectors $v_l$ for $1\leq l\leq L$ (including the case $n=2$ where we consider all vectors to be product). Let $W=\{w_j|1\leq j\leq J\}$ be a basis of a separable Hilbert space $\cH_n$ and consider the product bases $W^k=\{w_j\otimes v_l|w_j\in W, v_l\in V^k\}$ of $\cH_n\otimes\cH$ for each $1\leq k\leq K$. Now assume there is no unentangled KS theorem in $\cH_n\otimes\cH$ and, therefore, there exists a KS-colouring $c$ of the bases $W^k$. Consider the map $c'(v_l)=\sum_jc(w_j\otimes v_l)$ on the elements of the bases $V^k$. We will show this map is a KS-colouring of the bases $V^k$ and thus the assumption that there is no unentangled KS theorem in $\cH_n\otimes\cH$ must be false.
	
	Firstly, we have that $c$ assigns one to at most one of the vectors in $\{w_j\otimes v_l|1\leq j\leq J\}$ and assigns zero to the rest, since $c$ is a KS-colouring and the vectors are mutually orthogonal. Therefore $c'(v_l)\in\{0,1\}$ for all $1\leq l\leq L$. Secondly, if $\langle v_l,v_{l'}\rangle=0$ then the vectors $\{w_j\otimes v_l, w_j\otimes v_{l'}|1\leq j\leq J\}$ are mutually orthogonal and $c$ assigns one to at most one of the vectors. It follows that $c'(v_l)+c'(v_{l'})\leq 1$.  Finally, for all $W^k$ we have $c(w\otimes v)=1$ for some $w\in W$ and $v\in V^k$. Therefore, $c'(v)=\sum_jc(w_j\otimes v)=1$ for some $v\in V^k$ for all $1\leq k\leq K$.
	
	Since there exists a KS set in any separable Hilbert space of dimension at least three, the desired result follows by induction (and the irrelevance of the order of the Hilbert spaces in the tensor product). 
\end{proof}
\end{appendix}

\bibliographystyle{unsrturl}

\end{document}